\documentclass[11pt,letterpaper]{article}
\usepackage{amsmath,amsthm,nicefrac}
\usepackage{amsfonts, amstext}

\usepackage{comment}

\usepackage{typearea}
\typearea{14}
\usepackage[font={small,it}]{caption}

\usepackage{setspace}
\usepackage[compact]{titlesec}

\usepackage{enumitem}
\usepackage[breaklinks]{hyperref}
\hypersetup{colorlinks=true,%
            citebordercolor={.6 .6 .6},linkbordercolor={.6 .6 .6},%
citecolor=blue,urlcolor=black,linkcolor=blue}

\usepackage[nameinlink]{cleveref}
\Crefname{algocf}{Algorithm}{Algorithms}
\crefname{algocfline}{line}{lines}
\Crefname{invariant}{Invariant}{Invariants}
\Crefname{claim}{Claim}{Claims}
\Crefname{subclaim}{Subclaim}{Subclaims}

\usepackage{epsfig}
\usepackage{amsthm,amssymb}
\usepackage{amsthm,amssymb}

\usepackage{mathrsfs}
\usepackage{xspace}
\usepackage{soul}
\usepackage{latexsym}
\usepackage{bm}

\usepackage{framed}

\usepackage[dvipsnames]{xcolor}
\definecolor{DarkGray}{rgb}{0.66, 0.66, 0.66}
\definecolor{DarkPowderBlue}{rgb}{0.0, 0.2, 0.6}
\definecolor{fluorescentyellow}{rgb}{0.8, 1.0, 0.0}

\usepackage[ruled,vlined,linesnumbered,algonl]{algorithm2e}
\SetEndCharOfAlgoLine{}
\SetKwComment{Comment}{\footnotesize$\triangleright$\ }{}

\SetCommentSty{mycommfont}

\usepackage{thmtools,thm-restate}

\makeatletter
\allowdisplaybreaks
\makeatother

\newcounter{note}[section]

\sethlcolor{fluorescentyellow}

\newcommand{\initOneLiners}{%
    \setlength{\itemsep}{0pt}
    \setlength{\parsep }{0pt}
    \setlength{\topsep }{0pt}
}

\newtheorem{theorem}{Theorem}[section]
\newtheorem{lemma}[theorem]{Lemma}
\newtheorem{claim}[theorem]{Claim}

\newtheorem{corollary}[theorem]{Corollary}

\theoremstyle{definition}
\newtheorem{definition}[theorem]{Definition}

\theoremstyle{remark}
\newtheorem{remark}[theorem]{Remark}

\makeatletter 
\renewcommand{\theinvariant}{(I\@arabic\c@invariant)}
\makeatother

\newcommand{\eps}{\varepsilon}

\newcommand{\opt}{{\textsf{\textup{opt}}}\xspace}

\renewcommand{\emptyset}{\varnothing}

\newcommand{\E}{\mathbb{E}}

\newcommand{\junk}[1]{}
\newcommand{\eat}[1]{}

\newif\ifhideproofs

\ifhideproofs
\usepackage{environ}
\NewEnviron{hide}{}

\fi

\usepackage{graphicx} %
\usepackage{amsmath,amssymb,amsthm}
\usepackage{mathtools}
\usepackage{xspace}
\usepackage{fullpage}
\usepackage{mdframed}
\usepackage{hyperref}
\usepackage{enumitem}
\usepackage{booktabs,array}

\newcommand{\PP}{\mathbb{P}}  %
\newcommand{\MM}[1]{\mathrm{MM}\!\left(#1\right)}
\DeclareMathOperator{\Var}{Var}

\newcommand{\RS}{{{\rm RS}}}

\newcommand{\Gstar}{G^{\star}}
\newcommand{\Fstar}{F^{\star}}

\newcommand{\VC}[1]{\mathrm{VC}\!\left(#1\right)}
\newcommand{\MVC}[1]{\mathrm{MVC}\!\left(#1\right)}

\newcommand{\SEED}{\mathrm{SEED}}

\newcommand{\rest}{\mathrm{rest}}

\newcommand{\decided}[0]{\mathsf{decided}}
\newcommand{\undecided}[0]{\mathsf{undecided}}
\newcommand{\problematic}[0]{\mathsf{revealing}}

\newcommand{\OPT}{\mathrm{OPT}}

\newcommand{\abs}[1]{\left\lvert #1 \right\rvert}
\newcommand{\paren}[1]{\left( #1 \right)}

\newcommand{\vertexseedalg}{{\rm \textsc{Vertex-Seed}}\xspace}
\newcommand{\ouralg}{{\rm \textsc{Vertex-Cover}}\xspace}

\usepackage[dvipsnames]{xcolor}
\ifdefined\DEBUG

\newcommand{\ali}[1]{\textcolor{Red}{\sf{[Ali: #1]}}}
\newcommand{\inge}[1]{\textcolor{cyan}{\sf{[Inge: #1]}}}
\newcommand{\jan}[1]{\textcolor{red}{Jan: #1}}
\newcommand{\cliff}[1]{\textcolor{purple}{\sf{[Cliff: #1]}}}

\newcommand{\ola}[1]{\textcolor{Green}{Ola: #1}}

\fi

\newcommand{\thesymbol}{P}
\newcommand{\detsetvar}{\thesymbol}
\newcommand{\detset}{\hat{\thesymbol}}
\newcommand{\nonadapt}{\SEED_{\rm NA}}
\newcommand{\adapt}{\SEED_{\rm A}}

\newcommand{\mvcsetalg}{H}

\title{An Optimal Algorithm for Stochastic Vertex Cover}
\author{
Jan van den Brand\thanks{Georgia Institute of Technology. Email: \texttt{vdbrand@gatech.edu}}
\and
Inge Li G{\o}rtz\thanks{Technical University of Denmark. Email: \texttt{inge@dtu.dk}}
\and
Chirag Pabbaraju\thanks{Stanford University. Email: \texttt{cpabbara@stanford.edu.}}
\and
Debmalya Panigrahi\thanks{Duke University. Email: \texttt{debmalya@cs.duke.edu}}
\and
Clifford Stein\thanks{Columbia University. Email: \texttt{cliff@ieor.columbia.edu}}
\and
Miltiadis Stouras\thanks{EPFL. Email: \texttt{miltiadis.stouras@epfl.ch}}
\and
Ola Svensson\thanks{EPFL. Email: \texttt{ola.svensson@epfl.ch}}
\and
Ali Vakilian\thanks{Virginia Tech. Email: \texttt{vakilian@vt.edu}}
}
\date{}

\begin{document}

\maketitle

\thispagestyle{empty}

\begin{abstract}
    The goal in the stochastic vertex cover problem is to obtain an approximately minimum vertex cover for a graph $\Gstar$ that is realized by sampling each edge independently with some probability $p\in (0, 1]$ in a base graph $G = (V, E)$.
    The algorithm is given the base graph $G$ and the probability $p$ as inputs, but its only access to the realized graph $\Gstar$ is through queries on individual edges in $G$ that reveal the existence (or not) of the queried edge in $\Gstar$. In this paper, we resolve the central open question for this problem: to find a $(1+\eps)$-approximate vertex cover using only $O_\eps(n/p)$ edge queries. Prior to our work, there were two incomparable state-of-the-art results for this problem: a $(3/2+\eps)$-approximation using $O_\eps(n/p)$ queries (Derakhshan, Durvasula, and Haghtalab, 2023)  and a $(1+\eps)$-approximation using $O_\eps((n/p)\cdot \RS(n))$ queries (Derakhshan, Saneian, and Xun, 2025), where $\RS(n)$ is known to be at least $2^{\Omega\left(\frac{\log n}{\log \log n}\right)}$ and could be as large as $\frac{n}{2^{\Theta(\log^* n)}}$. Our improved upper bound of $O_{\eps}(n/p)$ matches the known lower bound of $\Omega(n/p)$ for any constant-factor approximation algorithm for this problem (Behnezhad, Blum, and Derakhshan, 2022). A key tool in our result is a new concentration bound for the size of minimum vertex cover on random graphs, which might be of independent interest.

\end{abstract}

\newpage

\setcounter{page}{1}

\maketitle

\section{Introduction}\label{sec:intro}In the {\em stochastic vertex cover} problem, we are given a {\em base graph} $G = (V, E)$ and a sampling probability $p\in (0, 1]$. The {\em realized graph} $\Gstar$ is a random graph generated by sampling each edge in $G$ independently with probability $p$. The algorithm does not have access to $\Gstar$ directly, but can query individual edges $e\in E$ to learn whether they appear in $\Gstar$. The goal of the algorithm is to output a near-optimal vertex cover of $\Gstar$ while querying as few edges in $G$ as possible.

The stochastic setting has been widely considered in graph algorithms. Perhaps the most extensive literature exists for the stochastic matching problem, where the goal is to query a sparse subgraph $H$ of the base graph $G$ such that the realized edges in $H$ contain an approximately maximum matching of the realized graph $\Gstar$. For this problem, the first result was obtained by 
Blum et~al.~\cite{blum2015ignorance}, %
who gave a $1/2$-approximation algorithm using $n\cdot \mathrm{poly}(1/p)$ queries. This result has subsequently been improved in an extensive line of work~\cite{assadi2016stochastic,assadi2017stochastic,behnezhad2018almost,behnezhad2019stochastic,assadi2019towards,behnezhad2020weighted,behnezhad2020stochastic,derakhshan2025matching}, eventually culminating in an (almost tight) result that gives a $(1-\eps)$-approximation using $n\cdot \mathrm{poly}(1/p)$ queries, for any fixed small $\eps>0$~\cite{azarmehr2025stochastic}.
Besides the maximum matching problem, stochastic optimization on graphs has also been considered for other classical problems such as minimum spanning tree and shortest paths~\cite{bertsekas1991analysis,goemans2006covering,vondrak2007shortest,blum2013harnessing}, as well as for more general frameworks in combinatorial optimization such as covering and packing problems~\cite{dean2008approximating,asadpour2008stochastic,golovin2011adaptive,bhalgat2011improved,yamaguchi2018stochastic}. 

In this paper, we consider the stochastic vertex cover problem. 
This problem was introduced by Behnezhad, Blum, and Derakhshan~\cite{behnezhad2022stochastic}, who gave, for any $\eps > 0$, a $(2+\eps)$-approximation (polynomial-time) algorithm using $O\left(\frac{n}{\eps^3 p}\right)$ queries. They also showed a simple lower bound that $\Omega(n/p)$ queries are necessary for any constant-factor approximation. The latter result is information-theoretic and rules out algorithms with fewer queries, even allowing an arbitrarily large running time. These results raised the question of whether there exist (exponential-time) algorithms that obtain a better-than-$2$ approximation, while still querying only $O(n/p)$ edges. This question was answered in the affirmative by Derakhshan, Durvasula, and Haghtalab~\cite{derakhshan2023stochastic}, who obtained an approximation factor of $\left(\frac{3}{2}+\eps\right)$, while querying $O\left(\frac{n}{\eps p}\right)$ edges. This result, which helped delineate the information complexity of the problem by breaching the polynomial-time solvability barrier, led to the natural question: can we obtain a near-optimal $(1+\eps)$-approximation algorithm using $O_\eps(n/p)$ queries? 

Interestingly, Behnezhad, Blum, and Derakshan had previously addressed this question for {\em bipartite} graphs: they obtained an $(1+\eps)$-approximation using $O_{\eps, p}(n)$ queries, although the dependence on $1/p$ and $1/\eps$ were triple-exponential~\cite{behnezhad2022stochastic}. 
The first algorithm to obtain a $(1+\eps)$-approximation for general graphs was obtained recently by Derakhshan, Saneian, and Xun~\cite{derakhshan2025query}, but their algorithm uses (super-linear) $O\left(\frac{n}{p}\cdot \RS(n)\right)$ queries. Here, $\RS$ refers to Ruzsa-Szemer\'edi Graphs, and $\RS(n)$ is the largest $\beta$ such that there exists an $n$-vertex graph whose edges can be partitioned into $\beta$ induced, edge-disjoint matchings of size $\Theta(n)$. 
The value of $\RS(n)$ is known to be at least $2^{\Omega\left(\frac{\log n}{\log \log n}\right)}$ and could be as large as $\frac{n}{2^{\Theta(\log^* n)}}$.

Although the number of queries in this last result is super-linear, it applies to a more general setting. Previously, \cite{derakhshan2023stochastic} had observed that under a regime permitting ``mild'' correlation between edge realizations, surpassing the $\frac{3}{2}$ factor requires $\Omega\left(n\cdot \mathrm{RS}(n)\right)$ queries. Since the result of \cite{derakhshan2025query} applies to this regime as well, their bound is tight in this mildly correlated setting. This left open the question of whether one can obtain a tight bound of $O_\eps(n/p)$ in the original setting of independent edge sampling, or whether this dependence on the parameter $\RS(n)$ was fundamental to the problem even with full independence.

\subsection{Our Result}

In this paper, we give a $(1+\eps)$-approximation algorithm for the stochastic vertex cover problem using $O_\eps(n/p)$ queries, for any small $\eps > 0$. By the lower bound of \cite{behnezhad2022stochastic}, our algorithm is {\em optimal}, up to the dependence on $\eps$ which is $1/\eps^5$. To the best of our knowledge, the only previously known $(1+\eps)$-approximation with linearly many (in $n$) queries was for bipartite graphs, but a linear dependence on $1/p$ (or polynomial dependence on $1/\eps$) was not known even in this special case. 
Our result also shows that the dependence on $\RS(n)$ in \cite{derakhshan2025query} was an artifact of the correlations between edges, and in this sense, is not fundamental to the stochastic vertex cover problem itself.

\begin{theorem}\label{thm:main}
    For any $\eps \in (0, c)$, where $c > 0$ is a small enough constant, there is a deterministic algorithm for the stochastic vertex cover problem that achieves an approximation factor of $1+\eps$ using $O\left(\frac{n}{\eps^5 p}\right)$ edge queries.
\end{theorem}

Similar to all previous algorithms for stochastic vertex cover, our approximation factor is with respect to the {\em expected} size of the minimum vertex cover in the realized graph $\Gstar$, and the set of edges queried by our algorithm is {\em non-adaptive}, i.e., independent of the realized graph $\Gstar$.

\subsection{Our Techniques}

\paragraph{Description of the algorithm.} Our algorithm is simple, and in some sense, canonical among non-adaptive algorithms. Since the set of edges queried by the algorithm is non-adaptive, the output must deterministically include a valid vertex cover on the remaining (non-queried) edges in $G$ to ensure feasibility. Call this deterministic vertex set $\detsetvar$. Since $\detsetvar$ is always part of the output, it is wasteful to query any edge incident to $\detsetvar$. So, we query the edges that are not incident to $\detsetvar$, i.e., that are in the induced graph $G[V\setminus \detsetvar]$. These queries reveal the realized graph $\Gstar[V\setminus \detsetvar]$; we compute its minimum vertex cover and add these vertices to $\detsetvar$ in the output. 
Finally, to choose $\detsetvar$ optimally, we solve an optimization problem that minimizes the expected size of the vertex cover output by the above algorithm, under the constraint that the number of edges in $G[V\setminus \detsetvar]$ is at most $O_\eps(n/p)$, our desired query bound. This optimal choice of $\detsetvar$ is denoted $\detset$. We call this the \ouralg algorithm and formally describe it in Section~\ref{sec:ouralg-description}.

\paragraph{An adaptive algorithm as an analysis tool.}
While our algorithm is simple, its analysis is quite subtle. 
In the rest of this section, we give an outline of the main ideas in the analysis. Since our algorithm is defined via an optimization problem, it is difficult to directly compare its cost to \opt, the expected size of an optimal vertex cover. Instead, we first define a surrogate algorithm that {\em adaptively} chooses a set $\SEED(\Gstar)$ such that $G[V\setminus \SEED(\Gstar)]$ has $O_\eps(n/p)$ edges. Later, we will compare this adaptive strategy to our non-adaptive \ouralg algorithm. The set $\SEED(\Gstar)$ has two parts: a non-adaptive part $\nonadapt$ and an adaptive part $\adapt(\Gstar)$. First, we describe the choice of $\nonadapt$. Let $\MVC{H}$ denote a minimum vertex cover of any graph $H$. 
We partition the vertices into three groups, $L$, $M$ and $S$, according to their probability
to appear in $\MVC{\Gstar}$. Vertices in $L$ have ``large'' probability (at least $1-2\eps$),
those in $M$ have ``moderate'' probability (between $\eps$ and $1-2\eps$) and the ones in $S$
have ``small'' probability (at most $\eps$).
Observe that we can safely include all vertices in $L$ to $\nonadapt$. That is because
$(1-2\eps)\cdot |L| \ge \opt$, from which it follows that $\E[|L \setminus \MVC{\Gstar}|] \le O(\eps)\cdot \opt$. Next, we can discard all vertices in $S$; these vertices will not appear in $\SEED(\Gstar)$ for any $\Gstar$. Using the fact that the vertices in $S$ are infrequent in the optimal solution, we show that there are only $O_\eps(n/p)$ edges that are entirely within $S$ or between $S$ and $M$, therefore these edges can be queried (notice that the edges between $S$ and $L$ will be covered by vertices in $L$). This allows us to focus on the subgraph $G[M]$ in the remaining discussion.

\paragraph{Deciding the adaptive part of $\SEED(\Gstar)$.}
The remaining vertices, namely the set $M$, appear in $\MVC{\Gstar}$ with probabilities ranging from $\eps$ to $1-2\eps$. We cannot afford to add all vertices in $M$ to $\nonadapt$, but we cannot discard all these vertices either. Instead, we use a greedy algorithm to select vertices in $M$ to add to $\nonadapt$. A natural strategy would be to choose vertices with high degree in $G$. Indeed, if the degree of a vertex is at most $O_\eps(1/p)$, we can query all its incident edges, and hence, it is redundant to add such a vertex to $\SEED(\Gstar)$. But, in general, high-degree vertices in $M$ may only appear in $\MVC{\Gstar}$ with probability $\eps$, and as such, they might be numerous compared to $\opt$. Since we cannot afford to add all these vertices to $\nonadapt$, we ask: which high-degree vertices should we prefer?

The answer to the above question lies at the heart of our analysis. Observe that if a high-degree vertex $v$ is {\em not} in $\MVC{\Gstar}$, then {\em all the neighbors of $v$} must be in $\MVC{\Gstar}$. Using this observation, we define a deterministic procedure (called the \vertexseedalg algorithm) that outputs a sequence of vertices $Q$, where the $i$-th vertex $v_i$ has the following property: with probability at least some constant $\delta$ over the choice of $\Gstar$, the vertex $v_i$ is not in $\MVC{\Gstar}$ {\em and} has a large neighborhood in $G$ among vertices whose status (whether in $\MVC{\Gstar}$ or not) has not been ``decided'' by previous vertices in $Q$. The intuition is that such a vertex {\em reveals} a large number of previously undecided vertices to be in $\MVC{\Gstar}$, and therefore, allows $|Q|$ to be bounded against $\opt$. We add the vertices in $Q$ to $\nonadapt$, and show that the expected size of $|Q \setminus \MVC{\Gstar}|$ is at most $O(\eps^2)\cdot \opt$ (a bound of $O(\eps)\cdot \opt$ would suffice for now, but later, we will need the sharper bound of $O(\eps^2)\cdot \opt$). Furthermore, we add the neighbors of vertices in $Q\setminus \MVC{\Gstar}$ to the adaptive set $\adapt(\Gstar)$; since these vertices must be in $\MVC{\Gstar}$, this does not affect the approximation bound. 

Finally, we consider the vertices $A$ that reveal a large number of previously undecided vertices to be in $\MVC{\Gstar}$ for a specific realization $\Gstar$, but are not in $Q$ because they do not meet the probability threshold $\delta$ over the different realizations of $\Gstar$. We add the vertices in $A$ to the adaptive set $\adapt(\Gstar)$ as well, and show that the expected size of $A\setminus \MVC{\Gstar}$ is also $O(\eps)\cdot \opt$. This completes the description of the set $\SEED(\Gstar)$. This last step ensures that the vertices outside $\SEED(\Gstar)$ each have only $O_{\eps}(1/p)$ edges in $G[M]$ that need to be queried.

\paragraph{Using $\SEED(G^*)$ to analyze our algorithm.} So far, we have described the adaptive set $\SEED(\Gstar)$, and outlined intuition for two facts: (1) that $\SEED(\Gstar)$ contains at most $(1+O(\eps))\cdot \opt$ vertices in expectation, and (2) that there are at most $O_\eps(n/p)$ edges in $G$ that are not covered by $\SEED(\Gstar)$. We remark that although $\SEED(\Gstar)$ satisfies the conditions of the optimization problem in the \ouralg algorithm, it does not immediately give an adaptive algorithm with $O_\eps(n/p)$ queries. This is because the computation of the set $\SEED(\Gstar)$ can require more than $O_\eps(n/p)$ queries. So, the reader should view the definition of $\SEED(\Gstar)$ strictly as an analysis tool, and not an alternative adaptive algorithm. Intuitively, we want to compare $\SEED(\Gstar)$ to $\detset$, the vertices chosen non-adaptively in the \ouralg algorithm. If $\SEED(\Gstar)$ were non-adaptive, this comparison is immediate since it would be a valid choice of $S$ in the optimization problem defining the \ouralg algorithm. But, in general, $\SEED(\Gstar)$ can vary based on the realization of $\Gstar$, and therefore, the expected size of $\SEED(\Gstar)$ might be smaller than $|\detset|$.

\paragraph{Dealing with the adaptivity of $\SEED(\Gstar).$}
Note that $\SEED(\Gstar)$ contains two parts: a non-adaptive set $\nonadapt$ and an adaptive set $\adapt(\Gstar)$. The set $\adapt(\Gstar)$ only depends on two random quantities: the identity of the set $Q\cap \MVC{\Gstar}$ and the realizations of edges incident to vertices in $Q\setminus \MVC{\Gstar}$. 
Importantly, our choice to include $Q$ in $\nonadapt$ makes $\SEED(\Gstar)$'s extension to a valid vertex cover (i.e. $\MVC{\Gstar[V\setminus\SEED(\Gstar)]}$) independent of the realizations of the edges incident
to $Q\cap\MVC{\Gstar}$. This allows us to fix the realization of these edges, and analyze $\SEED(\Gstar)$ over the remaining randomness in $\Gstar$. This limits the adaptivity of $\SEED(\Gstar)$, as it now only depends on the set $Q\cap\MVC{\Gstar}$ which can take at most $2^{|Q|} = 2^{O(\eps^2)\cdot \opt}$ values.
In addition, we show that the size of the $\SEED(\Gstar)$'s extension, namely $|\MVC{\Gstar[V\setminus\SEED(\Gstar)]}|$, sharply concentrates. We do so by proving a general theorem on the convergence of $|\MVC{\Gstar}|$ on a randomly generated graph $\Gstar\sim G_p$:

\begin{restatable}{theorem}{concentration}\label{thm:MVC-concentration}
    Let $Z = |\MVC{\Gstar}|$ and $\opt = \mathbb{E}_{\Gstar\sim G_p}[\abs{\MVC{\Gstar}}]$. 
    Then for any $t \geq 0$,
    \begin{equation} \label{eq:concentration_bernstein}
        \Pr[\abs{Z - \opt} \ge t] \leq 2 \exp\paren{-\frac{t^2}{4C\cdot\opt + 2t/3}},
    \end{equation}
    where $C < 8$ is a constant.
\end{restatable}

Note that this theorem is not specific to our construction and establishes a general 
concentration result for minimum vertex cover. The tail bound proven is much sharper than
what one can obtain from standard techniques (e.g., vertex exposure martingales) and we believe it 
might be of independent combinatorial interest.\footnote{A strengthened version of Talagrand's inequality for $c$-Lipschitz, $s$-certifiable functions also yields a slightly weaker concentration bound, which still suffices to derive our results; see~\Cref{rem:talagrand-mvc}.} 

Finally, we use a union bound over the $2^{O(\eps^2)\cdot \opt}$ different realizations of $Q\cap\MVC{\Gstar}$, for each of which the tail bound on the size of $\MVC{\Gstar[V\setminus \nonadapt]}$ applies. Using this, we can now claim that the advantage of adaptivity in defining $\SEED(\Gstar)$ is negligible. Formally, we show that, for any fixed realization of the edges incident to $Q$, the set $\SEED(\Gstar)$ can take at most $2^{O(\eps)\cdot \opt}$ different forms
for the graphs $\Gstar$ that are consistent with the fixed realization of edges incident to $Q$. Out of those forms,
the one that minimizes the expected size of $\SEED(\Gstar) \cup \MVC{\Gstar[V\setminus \SEED(\Gstar)]}$ is at most $(1+O(\eps))$ worse than adaptively selecting the best of these forms
for each graph $\Gstar$. Since the latter holds for any realization of the edges incident to $Q$,
averaging over their randomness gives us that there exists a deterministic set $\detsetvar$ that
produces a solution of expected size $(1+O(\eps))\opt$. This, in turn, establishes that our algorithm, which optimizes over all non-adaptive choices of $\detsetvar$, is a $(1+O(\eps))$-approximation algorithm.

\paragraph{Roadmap.} 
We formally define the stochastic vertex cover problem and give our \ouralg algorithm in Section~\ref{sec:ouralg}. The analysis of this algorithm, assuming Theorem~\ref{thm:MVC-concentration} (the concentration result for minimum vertex cover), in given in Section~\ref{sec:vertex-seed}. Finally, the proof of Theorem~\ref{thm:MVC-concentration} is given in Section~\ref{sec:mvc-concentration}.

\section{The \ouralg Algorithm}\label{sec:ouralg}In this section, we first formally define the stochastic vertex cover problem, and establish notation that we will use throughout the paper. Then, we give a formal description of our \ouralg algorithm.

\subsection{Notation and Terminology}
Throughout, let $G=(V,E)$ denote the base graph, with $n = |V|$. 
Fix an edge-realization parameter $p \in (0,1]$; each edge $e \in E$ is realized independently with probability $p$. 
All results extend to the heterogeneous model with edge-wise probabilities $(p_e)_{e\in E}$ by replacing $p$ in the statements with $p := \min_{e\in E} p_e$. For clarity of exposition, throughout the paper we present the homogeneous case $p_e = p$ for all $e \in E$.

Let $\Gstar$ be the random realized subgraph obtained by including each $e\in E$ independently with probability $p$; we write $\Gstar \sim G_p$. For a set of edges $F \subseteq E$, we use $\Gstar \setminus \Fstar$ to denote the distribution obtained by realizing all edges in $E \setminus F$ as above while leaving the edges in $F$ unresolved (to be realized later). An \emph{edge query} reveals whether a particular $e \in E$ is present in $\Gstar$.

For a graph $H=(V,E)$ and $v\in V$, let $N_H(v)$ be the set of neighbors of $v$ in $H$; for $S\subseteq V$, let
$N_H(S)=\{u\in V\setminus S: \exists s\in S\text{ with } (u,s) \in E\}$.
Let $\MVC{\cdot}$ be the mapping that assigns to every realized graph $\Gstar$ an arbitrary but fixed minimum vertex cover $\MVC{\Gstar}=S\subseteq V$. %
Let $\MVC{\Gstar}$ denote the resulting (random) minimum vertex cover on the realized graph, and define $\opt = \E_{\Gstar\sim G_p}\big[\abs{\MVC{\Gstar}}\big]$. %
For each $v \in V$, set $c_v = \Pr\big[v \in \MVC{\Gstar}\big]$, so that $\sum_{v \in V} c_v = \opt$. For an edge $e=(u,v) \in E$ define the probability that edge $e$ is covered as $c_e = \Pr[u \in \MVC{\Gstar} \text{ or } v \in \MVC{\Gstar}]$.

\subsection{The Stochastic Vertex Cover Problem}
Given a base graph $G=(V,E)$ and a probability parameter $p$, in the stochastic vertex cover problem, the goal is to output a feasible vertex cover of the realized graph $\Gstar$, in which each edge of $G$ is realized in $\Gstar$ independently with probability $p$, while querying as few edges in $G$ as possible. The algorithm is allowed unlimited access to the base graph $G$ as well as unlimited computation time. 
For an $\alpha>0$, we say a (randomized) solution $C$ is an $\alpha$-approximate stochastic vertex cover, if any edge in $\Gstar$ has at least one of its endpoint in $C$, and $\E[\abs{C}]\le \alpha \cdot \E[\abs{\MVC{\Gstar}}]$.

In this paper, we give a non-adaptive $(1+\eps)$-approximation algorithm for the stochastic vertex cover problem, i.e., it queries a fixed set of edges chosen in advance, independent of all query outcomes.

\subsection{Description of the \ouralg algorithm}\label{sec:ouralg-description}

First, we choose a subset $\detsetvar\subseteq V$, minimizing $\abs{\detsetvar} + \E\big[\abs{\MVC{\Gstar[V\setminus \detsetvar]}}\big]$ under the constraint that the induced subgraph $G[V\setminus \detsetvar]$ contains at most $O(n/(\eps^5p))$ edges. Let that set be $\detset$. We then query all edges in the induced subgraph $G[V\setminus \detset]$ and compute the minimum vertex cover $\mvcsetalg$ of the now known, realized graph $\Gstar[V\setminus \detset]$. Finally, we return the set $\detset \cup \mvcsetalg$ as our vertex cover. Recall, as in previous papers, that we are not concerned with computational efficiency, but only the query complexity. %
The pseudocode for the algorithm is as follows:

\begin{center}
    \begin{minipage}{0.95\textwidth}
        \begin{mdframed}[hidealllines=true, backgroundcolor=gray!20]
            {\bf Algorithm:} \ouralg~ %
            \begin{enumerate}
                \item Let $\detset$ be the optimal solution to \begin{equation}\label{problem:best-S}
                \begin{aligned}
                \min_{\detsetvar \subseteq V}\quad
                 & |\detsetvar| \;+\; \mathbb{E}\!\left[\;\bigl|\,\MVC{\Gstar[V\setminus \detsetvar]}\bigr|\;\right] \\
                    \text{s.t.}\quad
                  & G[V\setminus \detsetvar]\; \text{ has at most }\; 
                  2\cdot\frac{10^3\;n}{\eps^{5}p}\;%
                  \text{edges.}
            \end{aligned}
            \end{equation}
                \item Query the edges in $G[V\setminus \detset]$ to get its realization $\Gstar[V\setminus \detset]$.
                \item Let $\mvcsetalg = \MVC{\Gstar[V\setminus \detset]}$.
                \item Return $\detset \cup \mvcsetalg$.
            \end{enumerate}
        \end{mdframed}
    \end{minipage}
\end{center}

It is immediate that this algorithm correctly produces a valid vertex cover:
\begin{claim}
    The output of the \ouralg\ algorithm is a vertex cover for $\Gstar$.
\end{claim}
\begin{proof}
    Each edge either has an endpoint in $\detset$ or is an edge in the induced subgraph $\Gstar[V\setminus \detset]$.  
\end{proof}

\section{Analysis of the \ouralg Algorithm}\label{sec:vertex-seed}
In this section we analyse the \ouralg algorithm using our surrogate algorithm. In the first two subsections we describe the \vertexseedalg algorithm and bound the size of the set $Q$ chosen by the algorithm. Next we describe how to choose the adaptive set $\SEED(\Gstar)$ and prove that it contains at most $(1+O(\eps))\cdot \opt$ vertices in expectation. Finally, we prove that there are at
most $O_\eps(n/p)$ edges in $G$ that are not covered by $\SEED(\Gstar)$, and analyze the performance of the \ouralg algorithm.   

Throughout the analysis, we assume that $\opt \geq c \cdot \log(1/\eps)/\eps^2$ for some constant $c$.
In Appendix~\ref{app:proofs-sec-4} we show that this assumption is without loss of generality:
if $\opt = O(\log(1/\eps)/\eps^2)$, then the base graph $G$ contains $O(n/(p\eps^3))$ edges and the
\ouralg algorithm can select $\detset = \emptyset$, query all of $G$ and obtain an exact solution.

\subsection{The \vertexseedalg algorithm}
In this section, we describe the \vertexseedalg algorithm,  a deterministic algorithm that returns a sequence of vertices $Q = (v_1, v_2, \ldots, v_k)$ in the base graph $G$ depending only on a vertex cover function $\VC{\Gstar}$ that maps every realization $\Gstar$ to a fixed feasible vertex cover of $\Gstar$. 
The intuition is that each vertex $v_i$ in $Q$ corresponds to a query of the type ``is $v_i \in \VC{\Gstar}$?'' If the answer is negative, i.e. $v_i \notin \VC{\Gstar}$,  then all neighbors of $v_i$ in $\Gstar$ are necessarily in $\VC{\Gstar}$. Our goal in selecting $Q$ is to keep $|Q|$ small while revealing a large number of vertices in $\VC{\Gstar}$ by virtue of being neighbors of vertices in $Q\setminus \VC{\Gstar}$. The \vertexseedalg algorithm describes a greedy procedure for incrementally constructing $Q$ with this purpose. As input, in addition to the base graph $G$, the \vertexseedalg algorithm takes two parameters, $\delta$ and $\gamma$, which will be defined later. 

Before defining the algorithm, we introduce some notation:
\begin{itemize}
    \item For a sequence of vertices $Q_i = (v_1, v_2, \ldots, v_i)$ and a fixed realization $\Gstar$, define 
    \begin{align}
        \label{eqn:decided-definition}
        \decided(Q_i,\Gstar) = \left\{v \in V \setminus Q_i: N_{\Gstar}(v) \cap (Q_i \setminus \VC{\Gstar}) \not= \emptyset\right\}.
    \end{align}
    In words, a vertex $v$ is in $\decided(Q_i,\Gstar)$ if it has a neighbor that is in $Q_i$ but not in $\VC{\Gstar}$. 
    Note that a vertex $v \in \decided(Q_i, \Gstar)$ necessarily belongs to $\VC{\Gstar}$, by virtue of feasibility of $\VC{\Gstar}$. 
    We further let $\undecided(Q_i, \Gstar) = (V \setminus Q_i) \setminus \decided(Q_i, \Gstar)$. Thus, the sets $\decided(Q_i,\Gstar)$ and $\undecided(Q_i,\Gstar)$ induce a partition of the vertices in $V \setminus Q_i$.
    
    \item For a sequence of vertices $Q_i = (v_1, v_2, \ldots, v_i)$ and a fixed realization $\Gstar$, define
    \begin{align}
        \label{eqn:problematic-definition}
          \problematic(Q_i, \Gstar) = \left\{v\in V \setminus Q_i : v \not \in \VC{\Gstar} \mbox{ and } |N_G(v) \cap \undecided(Q_i, \Gstar)| \geq  \frac{1}{p\cdot \gamma}  \right\}\,.
    \end{align}
    In words, a vertex $v$ is in $\problematic(Q_i, \Gstar)$ if it is not in $\VC{\Gstar}$ {\em and} has at least $\frac{1}{p\cdot \gamma}$ neighbors in the base graph $G$ that are in $\undecided(Q_i, \Gstar)$.
    Intuitively, a vertex $v \in \problematic(Q_i, \Gstar)$ is a good candidate for extending $Q_i$ to $Q_{i+1}$ since it is likely to move a large number of vertices from $\undecided(Q_i, \Gstar)$ to $\decided(Q_{i+1}, \Gstar)$.
\end{itemize}

The \vertexseedalg algorithm starts with $Q$ being empty and then constructs it iteratively by adding vertices that have a high probability of being in $\problematic(Q, \Gstar)$. The pseudocode is as follows:
\begin{center}
    \begin{minipage}{0.95\textwidth}
        \begin{mdframed}[hidealllines=true, backgroundcolor=gray!20]
            {\bf Algorithm:} \vertexseedalg~ \\[-2mm]

            Initialize $Q$ to be the empty sequence.\\[-2mm]

            While there is a vertex $v$ such that $\Pr_{\Gstar}[v\in \problematic(Q, \Gstar)] \geq \delta$, append $v$ to $Q$.
        \end{mdframed}
    \end{minipage}
\end{center}

When analyzing the algorithm,
we use the notation $Q_i = (v_1, \ldots, v_i)$ to denote the set of vertices in $Q$ after $i$ iterations of the \vertexseedalg algorithm: $Q_0$ denotes the initial empty set, and $Q_k$ denotes the final set if the algorithm terminates after $k$ iterations. Note that the sequence of computed $Q_i$'s is deterministic, and when we measure probability within the loop (namely $\Pr_{\Gstar}[v\in \problematic(Q_i, \Gstar)]$), it is only over the draw of $\Gstar$. Furthermore, note that the algorithm necessarily terminates after at most $n$ iterations.

\subsection{Bounding the number of vertices in \texorpdfstring{$Q$}{Q}}

Observe that the \vertexseedalg algorithm is deterministic and hence, the number of vertices in $Q$ is also deterministic. The next lemma bounds this number.

\begin{lemma}
\label{lem:vertexseed}
    Let $Q = (v_1, \ldots, v_k)$ be the sequence of vertices constructed by the \vertexseedalg algorithm. For any $n \ge 4\log(2/\delta)$, we have that $k \leq \frac{10 \gamma}{\delta} \cdot n$.  %
\end{lemma}
First, we give the intuition behind the lemma, before we formally prove it. Suppose a vertex $v\in \problematic(Q, \Gstar)$ is added to $Q$ by the \vertexseedalg algorithm. Then, it has $\frac{1}{p\cdot \gamma}$ neighbors in $G$ that are currently undecided. Informally, we should expect that $1/\gamma$ of these neighbors realize in $\Gstar$ and, therefore, are moved from $\undecided(Q, \Gstar)$ to $\decided(Q, \Gstar)$ as a result of adding $v$ to $Q$. Since there are only $n$ vertices in $G$, the number of times this can happen is $\gamma\cdot n$. Finally, since the event $v\in \problematic(Q, \Gstar)$ holds with probability $\delta$ for every vertex added to $Q$, we should expect that $Q$ can only contain $\frac{\gamma \cdot n}{\delta}$ vertices. 

The argument above is not formal because the expected number of vertices moved from $\undecided(Q, \Gstar)$ to $\decided(Q, \Gstar)$ by a vertex $v\in \problematic(Q, \Gstar)$ is not necessarily $1/\gamma$, since the neighborhood of $v$ in $\Gstar$ is not independent of the event $v\in \problematic(Q, \Gstar)$. Nevertheless, we show that this intuition holds, and can be made formal, in the proof below.

\begin{proof}[Proof of Lemma~\ref{lem:vertexseed}]
    For $i=1, \ldots, k$, let $X_i$ be the (random) indicator variable that is $1$ if $v_i \in \problematic(Q_{i-1}, \Gstar)$ and $0$ otherwise. Note that $\E[X_i] \ge \delta$ for all $i$ by definition of the \vertexseedalg algorithm.
    Let $X = \sum_{i=1}^k X_i$. Thus, $\E\left[ X \right] = \sum_{i=1}^k \mathbb{E}[X_i] \geq k \cdot \delta$. 
    
    Let $u = \frac{10\gamma}{\delta}\cdot n$. Recall that we want to show that $k \le u$. We consider two cases, $X \le u \cdot \delta/2$ and $X > u \cdot \delta/2$, and write 
    \begin{align}\label{eq:Q-bound}
        k\cdot \delta \leq \E\left[ X \right] \leq u\cdot \delta/2  + k \cdot \Pr[X > u\cdot \delta/2].
    \end{align}
    To complete the proof, we will show that $\Pr[X > u\cdot \delta/2] \leq \delta/2$, which implies $k \leq u$ as required.

Now let $\mathcal{G}$ be the subset of realizations of $\Gstar$ for which  $X \geq u\cdot \delta/2$. In other words, we want to show that $\Pr[\Gstar \in \mathcal{G}] \le \delta/2$.
 We further partition $\mathcal{G}$ as follows: for each $S \subseteq Q$ such that $|S| \ge u\cdot \delta/2$, we let $\mathcal{G}_S \subseteq \mathcal{G}$ be the realizations of $\Gstar$ for which $\{v_i \in Q \;|\; v_i \in \problematic(Q_{i-1}, \Gstar)\} = S$. %
    These chosen $\mathcal{G}_S$ partition $\mathcal{G}$ since, by definition of $\mathcal{G}$, every realization $\Gstar \in \mathcal{G}$ has 
    \[
    X = |\{v_i \in Q \mid v_i \in \problematic(Q_{i-1}, \Gstar)\} | \geq u \cdot \delta/2.
    \]
    We thus have that $\mathcal{G}$ is partitioned into at most $\sum_{\ell=u\delta/2}^k \binom{k}{\ell} \le 2^k$ many sets $\mathcal{G}_S$. %

    We proceed to bound $\Pr[\Gstar \in \mathcal{G}_S]$  for a fixed set $S$. To do this, the following viewpoint for generating a realization $\Gstar$ will be helpful. There is a random string $R = (r_1, r_2,\ldots)$ of bits where each $r_i$ is a random bit that is $1$ with probability $p$ and $0$ otherwise. The realization $\Gstar$ is generated as follows:
    \begin{itemize}
        \item For each $v_i \in S$ (in the order it was added to $Q$), the existence of an edge $e=(v_i, v)$ in $\Gstar$ is determined by the next random bit in $R$ if the following properties are satisfied by $v$:
        \begin{itemize}
        		\item $v$ is a neighbor of $v_i$ in the base graph $G$,
        		\item $v$ is not in $\{v_1, \ldots, v_{i-1}\}$, and
        		\item there is no edge $(v_j, v)$ realized so far in $\Gstar$ such that $v_j$ is a vertex in $\{v_1, \ldots, v_{i-1}\} \cap S$. 
        	\end{itemize}
        \item Any remaining edge not considered above is simply realized with probability $p$ independently.
    \end{itemize}

    The above procedure simply realizes each edge with probability $p$ independently but distinguishes (as a function of $S$) certain random bits with $R$.

    Now, we claim that if any $\Gstar$ realized this way ends up belonging to $\mathcal{G}_S$, then we must have used at least $5n/p$ bits from $R$. To see this, consider the realization of edges adjacent to some fixed vertex $v_i \in S$ in the process above. Since the realized graph $\Gstar \subseteq \mathcal{G}_S$, it must be the case that $v_i \in \problematic(Q_{i-1}, \Gstar)$, which  means that $v_i$ has at least $1/(p\cdot \gamma)$ neighbors in $G$ that are in $\undecided(Q_{i-1}, \Gstar)$. Since $\undecided(Q_{i-1}, \Gstar) \subseteq V \setminus Q_{i-1}$, all these neighbors are not in $Q_{i-1}$. Moreover, since each vertex $v_i\in S$ is in $\problematic(Q_{i-1}, \Gstar)$, we have that none of the vertices in $S$ are in $\VC{\Gstar}$. Therefore, if a vertex $v$ is in $\undecided(Q_{i-1}, \Gstar)$, its neighborhood in $\Gstar$ is disjoint from $Q_{i-1} \cap S$. As a result, $v_i$ has at least
    $1/(p\cdot \gamma)$ neighbors $v$ in $G$, such that $v \notin Q_{i-1}$, and furthermore, the neighborhood of $v$ in $\Gstar$ is disjoint from $Q_{i-1}\cap S$. %
    
    For any such neighbor $v$ of $v_i$, the process described above uses a new bit from $R$ to confirm the presence of the edge $(v_i, v)$. In total, the number of potential edges determined using the random bits from $R$ is hence at least
    \[
        \frac{|S|}{p\cdot \gamma}  \geq  \frac{u \cdot \delta/2}{p\cdot \gamma} = \frac{5 n}{p}, \text{ since } u = \frac{10\gamma}{\delta}\cdot n.
    \]
    
    Observe also that if an edge $(v_i,v)$ is confirmed to be in $\Gstar$ using a bit from $R$, then no other edge $(v_j, v)$ for $j > i$ is determined using a bit from $R$. %
    This means that the number of edges confirmed to be in $\Gstar$ using bits from $R$ can be at most the number of distinct vertices in $G$, which is $n$. In summary, we conclude that, if the realized graph $\Gstar \in \mathcal{G}_S$, then it must be the case that we used at least $5n/p$ bits from $R$, and at most $n$ of these bits realized edges in $\Gstar$. So, consider the first $5n/p$ bits in $R$. As the expectation of the number of ones in the independent Bernoulli trials $r_1, \ldots, r_{5n/p}$ is $5n$, we have by a standard Chernoff bound that
    \begin{align*}
        \Pr[\Gstar \in \mathcal{G}_S] \le \Pr\left[\sum_{j=1}^{5n/p}r_j \le \left(1-\frac{4}{5}\right)5n\right] \leq \exp\left(\frac{-5n\cdot(4/5)^2}{3}\right) \le e^{- n}.
    \end{align*}
    Finally, by a union bound over the partitioning of $\mathcal{G}$,
    \begin{align*}
        \Pr[\Gstar \in \mathcal{G}] \leq  2^k\cdot e^{-n}
 		\le (2/e)^{n}       
         \leq \delta/2, %
    \end{align*}
    where the last inequality holds for $n \ge 4\log(2/\delta)$.
    So, by Eq.~\eqref{eq:Q-bound}, $k \le \frac{10\gamma}{\delta}\cdot n$ as required.     
\end{proof}

\subsection{Selection of the vertex set \texorpdfstring{$\SEED(\Gstar)$}{SEED(G*)}}

We can use the \vertexseedalg algorithm directly on the base graph $G$ and add the vertices in $Q$ to the set $\SEED(\Gstar)$. But, this results in a set $Q$ whose size is independent of the expected size of $\MVC{\Gstar}$, which translates to an additive error in the approximation bound of the algorithm. Instead, we use the \vertexseedalg algorithm in a more nuanced fashion that avoids this additive loss.

We partition vertices $V$ in $G$ into three sets:
\begin{itemize}
	\item The set $L := \{v\in V: \Pr_{\Gstar}[v\in \MVC{\Gstar}] \ge 1-2\eps \}$ of vertices that have {\em large} probability of being in $\MVC{\Gstar}$, 
	\item the set $M := \{v\in V: \eps < \Pr_{\Gstar}[v\in \MVC{\Gstar}] < 1-2\eps \}$ of vertices that have {\em moderate} probability of being in $\MVC{\Gstar}$, and
	\item the set $S := \{v\in V: \Pr_{\Gstar}[v\in \MVC{\Gstar}] \le \eps \}$ of vertices that have {\em small} probability of being in $\MVC{\Gstar}$.	
\end{itemize}
Recall that $\opt$ denotes the expected size of the minimum vertex cover, i.e., $\opt := \E[\MVC{\Gstar}]$. In the next two claims, we bound the number of vertices in $L$ and $M$, in terms of $\opt$. Note that $L$ and $M$ are deterministic sets, but $\MVC{\Gstar}$ is random.
\begin{claim}\label{cl:L}
    For any $\eps \le 1/4$,  $\E[L\setminus \MVC{\Gstar}]$ is at most $4\eps\cdot \opt$. %
\end{claim}
\begin{proof}
    Note that for any vertex $v\in L$, we have $\Pr[v\in L\setminus \MVC{\Gstar}] \le 2\eps$. Therefore,
\[%
		\E[|L\setminus \MVC{\Gstar}|] 
		\le 2 \eps \cdot |L|.
\]
Using this bound, we get
\[
    |L| \le \E[\abs{L\setminus \MVC{\Gstar}}] + \E[\abs{\MVC{\Gstar}}] \le 2\eps\cdot |L| + \opt,
\]
which implies that $|L| \le \opt / (1-2\eps)$. Therefore,
\[
    \E[|L\setminus \MVC{\Gstar}|] 
		\le 2 \eps \cdot |L|
		\le \frac{2 \eps}{1-2\eps} \cdot \opt
		\le 4\eps\cdot \opt,
		\text{ for } \eps \le 1/4.\qedhere
\]      
\end{proof}

\begin{claim}\label{cl:M}
    The number of vertices in $M$ is at most $\opt / \eps$.
\end{claim}
\begin{proof}
    Every vertex in $M$ is in $\MVC{\Gstar}$ with probability at least $\eps$, and $\opt$ is at least the expected number of vertices in $M$ that are in $\MVC{\Gstar}$. The claim follows.
\end{proof}

Now, we run the \vertexseedalg algorithm on the induced graph $G[M]$ and designate the output set $Q$. We also use $\Gstar[M]$ to denote the induced subgraph on $M$ of the realized graph $\Gstar$.
In the \vertexseedalg algorithm, we use the following parameters:
\begin{itemize}
    \item[-] The parameter $\delta,$ which is used as the probability threshold for a vertex to be added to $Q$, is set to $\delta \coloneqq \eps^2$. 
    \item[-] The parameter $\gamma$, which decides the threshold on the number of undecided neighbors in the base graph $G$ for a vertex to be deemed $\problematic$, is set to $\gamma \coloneqq \eps^3 \cdot \delta/10^3 = \eps^5/10^3$.
\end{itemize}
 Furthermore, the vertex cover $\VC{\Gstar[M]}$ used in the \vertexseedalg algorithm is set to $\MVC{\Gstar} \cap M$. Clearly, this is a feasible vertex cover of $\Gstar[M]$.  

Based on this choice of parameters, we can bound the number of vertices in $Q$. Recall that \vertexseedalg is a deterministic algorithm and it is run on a deterministic graph $G[M]$. Hence, $Q$ is deterministic.
\begin{claim}\label{cl:Q}
    The size of the set $Q$ output by the \vertexseedalg algorithm when run on $G[M]$ is at most $(\eps^2/100) \cdot \opt$.
\end{claim}
\begin{proof}
    By Lemma~\ref{lem:vertexseed} and our choice of parameters, we get $|Q| \le (\eps^3/100)\cdot |M|$. Now, the claim follows from the bound on $M$ in Claim~\ref{cl:M}.
\end{proof}

We are now ready to define the vertex set $\SEED(\Gstar)$ for any fixed realized graph $\Gstar$. We define $\SEED(\Gstar)$ as the union of four sets described below. The first two sets, $L$ and $Q$, do not depend on the realized graph $\Gstar$, whereas the last two sets depend on $\Gstar$.
\begin{itemize}
    \item[-] the set $L$ of vertices that have large probability of being in $\MVC{\Gstar}$,
    \item[-] the set $Q \subseteq M$ of vertices returned by the \vertexseedalg algorithm on $G[M]$,
    \item[-] the set of vertices decided by $Q$ in $\Gstar[M]$, i.e., $\decided(Q, \Gstar[M])$, and
    \item[-] the set of vertices $A := \{v\in M \setminus Q: |N_{G[M]}(v) \cap \undecided(Q, \Gstar[M])|\ge \frac{1}{p\cdot \gamma}\}$.
\end{itemize}
Next we prove that in expectation, $\SEED(\Gstar)$ contains a small number of vertices that are not in $\MVC{\Gstar}$. To do this we first  prove the following bound on the size of $A$.

\begin{claim}\label{cl:A}
    We have $\E[|A \setminus \MVC\Gstar|]%
    \leq \eps\cdot \opt$.
\end{claim}
\begin{proof}
  Consider any vertex $v \in M \setminus Q$. Upon sampling $\Gstar \sim G_p$, observe that if $v \in A \setminus \MVC\Gstar$, %
  then $v \in \problematic(Q, \Gstar[M])$. By definition of the termination of the \vertexseedalg algorithm, we therefore have that \[
  \textrm{Pr}_{\Gstar}[v \in A \setminus \MVC\Gstar]\; \le \;\textrm{Pr}_{\Gstar}[v\in \problematic(Q, \Gstar[M])] < \delta.
  \]
  Finally, by the bound on $M$ in Claim~\ref{cl:M}, this implies $\E[|A \setminus \MVC\Gstar|]< \delta\cdot |M| \leq \eps\cdot\opt$.   
\end{proof}

We can now bound the expected size of the set $\SEED(\Gstar)$.

\begin{lemma}\label{lem:seed-value-is-1-plus-epsilon}
    We have $\E[|\SEED(\Gstar) \setminus \MVC{\Gstar}|] \leq O(\eps)\cdot \E[\MVC{\Gstar}]$ for any $\eps < 1/4$. %
\end{lemma}
\begin{proof}
Note first, that $\decided(Q, \Gstar[M])\subseteq \MVC\Gstar$, since by the definition of $\decided$, for each vertex $v\in \decided(Q, \Gstar[M])$,  there is a vertex $v_i \not \in \MVC{\Gstar}$ that is adjacent to $v$ in the realization $\Gstar$.
Thus,  
\[
\E[|\SEED(\Gstar)\setminus \MVC{\Gstar}|] = \E[|L\setminus\MVC{\Gstar}|] + \E[|Q\setminus\MVC{\Gstar}|] + \E[|A\setminus \MVC{\Gstar}|].\] 
By Claim~\ref{cl:L}, \ref{cl:Q}, and \ref{cl:A}, we have 
\[
    \E[|\SEED(\Gstar)\setminus \MVC{\Gstar}|] \leq 
    4\eps\cdot \opt + (\eps^2/100) \cdot \opt + \eps \cdot \opt \leq O(\eps)\cdot \E[|\MVC{\Gstar}|]. \qedhere
\]
\end{proof}

\subsection{Using \vertexseedalg to analyze \ouralg}

\newcommand{\optS}{\hat{\thesymbol}}

\noindent In this section we prove our main result about the performance of \ouralg algorithm,
which we formally state in the next theorem.

\begin{theorem}\label{thm:ouralg-1-plus-epsilon}
    The output of the \ouralg algorithm has an expected size of at most $(1+O(\eps))\cdot\opt$.
\end{theorem}

To analyze the \ouralg algorithm, we first define an auxiliary problem (Problem~\eqref{problem:best-S-with-Q}), which strengthens the constraints of Problem~\eqref{problem:best-S} by requiring that the solution $S$ includes the vertices in $Q$, where $Q$ is the deterministic output
of the \vertexseedalg algorithm on $G[M]$. For the remaining of the section we will use $\optS$ to denote the optimal solution to Problem~\eqref{problem:best-S-with-Q}.

\begin{equation}\label{problem:best-S-with-Q}
\begin{aligned}
\min_{\detsetvar \subseteq V}\quad
  & |\detsetvar| \;+\; \mathbb{E}\!\left[\;\bigl|\,\MVC{\Gstar[V\setminus \detsetvar]}\,\bigr|\;\right] \\
\text{s.t.}\quad
  & Q \subseteq \detsetvar, \\
  & G[V\setminus \detsetvar]\; \text{ has at most }\; 2\cdot\frac{10^3\;n}{\eps^{5}p}\; \text{edges.}
\end{aligned}
\end{equation}

Our analysis compares the cost of $\optS$ to the expected cost of a strategy that adaptively picks $S = \SEED(\Gstar)$ for every realization $\Gstar$. We begin by proving that for any fixed realization $\Gstar$, the set $\SEED(\Gstar)$ is a feasible solution to \eqref{problem:best-S-with-Q} as $Q\subseteq \SEED(\Gstar)$, by definition, and by showing that the induced graph $G[V\setminus\SEED(\Gstar)]$ is sparse (Lemma~\ref{lem:v-minus-seed-is-sparse}).
Moreover, Lemma~\ref{lem:seed-value-is-1-plus-epsilon} implies that $\E[|\SEED(\Gstar)| + |\MVC{\Gstar[V\setminus\SEED(\Gstar)]}|] = (1+O(\eps))\cdot\opt$, therefore adaptively picking $\SEED(\Gstar)$ as a solution for each $\Gstar$, would achieve an expected objective value of $(1+O(\eps))\cdot\opt$ in Problem~\eqref{problem:best-S-with-Q}.
Our goal is to prove that the optimal static solution $\optS$ performs essentially as well.

For any realization $\Gstar$, the set $\SEED(\Gstar)$ depends only on two quantities: (a) on the realization of edges incident to $Q$ and (b) on the intersection of $Q$ with $\MVC{\Gstar}$.
Let $F$ be the set of edges, of the base graph $G$, which have at least one endpoint in $Q$ and let also $\Fstar \subseteq F$ denote a fixed
realization of them. 
For convenience, we slightly abuse notation and index each possible $\SEED$ set by the tuple $(Q_{VC},\Fstar)$ that determines it.
For the sake of completeness, we re-state the formal definitions of $\decided$, $\undecided$ and $\SEED$ in this new notation. For a fixed $\Fstar \subseteq F$ and a fixed $Q_{VC} \subseteq Q$, we have

\begin{itemize}
    \item $\decided(Q_{VC}, \Fstar) = \{u \in M\setminus Q: \exists(u,v) \in \Fstar \text{ s.t. } v \in (Q\setminus Q_{VC})\}$
    \item $\undecided(Q_{VC}, \Fstar) = M \setminus (Q \cup \decided(Q_{VC}, \Fstar))$ 
    \item $A( Q_{VC}, \Fstar) = \{u \in M \setminus Q:\; |N_{G[M]}(u)\cap\undecided( Q_{VC}, \Fstar)| \geq 1/(p\cdot \gamma)\}$
    \item $\SEED(Q_{VC}, \Fstar) = L \cup Q \cup \decided(\Fstar, Q_{VC})\cup A(Q_{VC}, \Fstar)$.
\end{itemize}
\noindent Notice that for any realization $\Gstar$, for which $\Fstar$ is the realization of the edges in $F$, we have that $\decided(Q \cap \MVC{\Gstar}, \Fstar) = \decided(Q, \Gstar[M])$ and $\SEED(\Gstar) = \SEED(Q\cap\MVC{\Gstar}, \Fstar)$.

A key observation is that the objective of Problem~\ref{problem:best-S-with-Q} is independent of how the edges in $F$ realize, because every feasible solution $S$ contains the set $Q$, and thus
the edges in $F$ are not part of the induced graph $G[V\setminus S]$. We can therefore fix a realization $\Fstar$ of these edges and compare $\optS$ and $\SEED(\Gstar)$ over the randomness outside of $F$. For a fixed realization $\Fstar$, $\SEED$ only depends on $Q_{VC}$, which ranges over subsets of $Q$ and can, thus, take at most $2^{|Q|}$ values.

By employing the concentration of the minimum vertex cover (Theorem~\ref{thm:MVC-concentration})
and taking a union bound over the $2^{|Q|}$ possible values of $\SEED$, we are able to show that
the value of $\SEED$ that has the minimum expected cost (for the fixed $\Fstar$ and over the randomness outside of $F$) is almost as good as adaptively selecting $\SEED(\Gstar)$. Since each value of $\SEED$ is feasible for Problem~\ref{problem:best-S-with-Q} and the objective only depends on the randomness outside $F$, we can conclude that $\optS$ cannot be worse than the best fixed $\SEED$.

Finally, since the above holds for any fixed $\Fstar$, we conclude the proof taking the expectation over the randomness in $F$ and showing that the expected cost of $\optS$ almost as good as the expected
cost of $\SEED(\Gstar)$.

\vspace{1em} We will now formalize the proof outlined above. First, we begin by proving that for any $\Fstar \subseteq F$ and any $Q_{VC} \subseteq Q$, the induced
graph $G[V\setminus\SEED(Q_{VC}, \Fstar)]$ is sparse. This proves that $\SEED(Q_{VC}, \Fstar)$ is a feasible solution to Problem~\eqref{problem:best-S-with-Q}.

\begin{lemma}\label{lem:v-minus-seed-is-sparse}
    For any $\Fstar \subseteq F$ and any $Q_{VC} \subseteq Q$, the graph $G[V \setminus \SEED(Q_{VC}, \Fstar)]$ has at most $\frac{2\cdot10^3\cdot n}{p \eps^5}$ many edges.
\end{lemma}

\begin{proof}[Proof of Lemma~\ref{lem:v-minus-seed-is-sparse}]

Since $\SEED(Q_{VC}, \Fstar)$ contains $L$, all vertices in $V\setminus \SEED(Q_{VC}, \Fstar)$ lie in $(M \cup S) \setminus \SEED(Q_{VC}, \Fstar)$. Moreover, since $\SEED(Q_{VC}, \Fstar)$ also contains $\decided(Q_{VC}, F^\star)$,
\[
M\setminus \SEED(Q_{VC}, \Fstar) \;\subseteq\; \undecided(Q_{VC}, \Fstar).
\]

Now, every vertex \(u \in M \setminus \SEED(Q_{VC}, \Fstar)\) has at most \(1/(p\cdot\gamma)\) neighbors inside \(\undecided(Q_{VC}, \Fstar)\) (since otherwise, such a vertex would be included in $A(Q_{VC}, \Fstar)$). Hence the number of edges in
\(G[V\setminus \SEED(Q_{VC}, \Fstar)]\) that have both endpoints in \(M\) is at most
\[
\frac{|M|}{p\gamma} \;\leq\; \frac{10^3\;n}{p\eps^5},
\]
since \(\gamma=\eps^5/10^3\).

The remaining edges in \(G[V\setminus \SEED(Q_{VC}, \Fstar)]\), are either between $S$ and $M$ or entirely within $S$. Fix any such edge \(e=(u,v)\) with \(u\in S\) and \(v\in (M\cup S) \setminus \SEED(Q_{VC}, \Fstar)\), and let \(\MVC{\Gstar}\) be the minimum vertex cover of a realization \(\Gstar\). Let
\[
c_e \;:=\; \Pr[e \text{ is covered by } \MVC{\Gstar}]
  \;\le\; \Pr[u\in \MVC{\Gstar}]+\Pr[v\in \MVC{\Gstar}]
  \;\le\; \eps+1-2\eps =1-\eps.
\]

As proven in \cite{derakhshan2023stochastic}, there can be at most $O(n/(\eps\cdot p))$ such edges $(u,v)$ in $G$.
We also provide a concise proof here for completeness. Let \(W=\{e\in E: c_e\le 1-\eps\}\). Set \(X=|\{e\in E: e \text{ is \emph{not} covered by }\MVC{\Gstar}\}|\).
By linearity of expectation,
\begin{equation}\label{eq:EX-lb}
\E[X] \;=\; \sum_{e\in E}(1-c_e)\;\ge\;\sum_{e\in W}(1-c_e)\;\ge\;\eps\,|W|.
\end{equation}

To upper bound \(\E[X]\), observe that every uncovered edge under \(\MVC{\Gstar}\) must be absent from \(\Gstar\) (otherwise \(\MVC{\Gstar}\) would not be a vertex cover of \(\Gstar\)).
The probability that a vertex-induced subgraph of $G$ has more than $n/p$ unrealized edges is at most
$(1-p)^{n/p} \le e^{-n}$. Also, there are at most $2^n$ choices for the vertex induced subgraph $H = G[V\setminus \MVC{\Gstar}]$. Therefore,

$$\Pr\left[H \text{ has more than } n/p \text{ edges}\right] \leq (2/e)^{n}\;.$$

Using the above, we can upper bound  \(\E[X]\) as
\begin{equation}\label{eq:EX-ub}
\E[X]\;\le\; \frac{n}{p} + n^2\; \left(\frac{2}{e}\right)^n \leq \frac{n}{p} + 6.
\end{equation}

Combining \eqref{eq:EX-lb} and \eqref{eq:EX-ub} yields the deterministic bound
\[
|W| \;\le\; \frac{\E[X]}{\eps}\;\le\; \frac{n}{p\;\eps} + \frac{6}{\eps} \leq 2\;\frac{n}{p\;\eps} \leq 2\;\frac{n}{p\eps^5}.
\]
Therefore the total number of remaining edges is at most $2\cdot10^3\;n/(p\eps^5)$ which concludes the proof of the lemma.
\end{proof}

For convenience, we define $g(A)$, for any set $A \subseteq V$, to be the random variable corresponding to the size of a vertex cover of
$\Gstar$ that includes the set $A$ and covers the induced graph $\Gstar[V\setminus A]$ optimally, i.e.,
$$g(A) = |A| + |\MVC{\Gstar[V\setminus A]}|.$$

\noindent Notice that the objective values of the optimization problems \eqref{problem:best-S} and \eqref{problem:best-S-with-Q} are equal to $\E[g(S)]$.
Before we continue with the analysis, we will prove the following tail bound for $g(S)$, by using the concentration of the minimum vertex cover (Theorem~\ref{thm:MVC-concentration}).

\begin{lemma}\label{lem:concentration-g(s)}
    For any set $S \subseteq V$ and any $\eps \in [0,1]$, the following holds
    $$
    \Pr\left[\;|g(S) - \E[g(S)]| > \eps\; \E[g(S)]\;\right] \leq 2\;e^{-\eps^2\;\opt/66}.
    $$
\end{lemma}

\begin{proof}[Proof of Lemma~\ref{lem:concentration-g(s)}]

Let $X:=|\MVC{\Gstar[V\setminus S]}|$. Since $|S|$ is deterministic,
\begin{equation}\label{eq:split}
\Pr\!\left[\,|g(S)-\E[g(S)]|>\eps\,\E[g(S)]\,\right]
=\Pr\!\left[\,|X-\E[X]|>\eps\,\E[X]+\eps\,|S|\,\right]. \tag{$\ast$}
\end{equation}

\noindent\textbf{Case (a): }$\E[X] \ge \opt/2$.
From \eqref{eq:split},
\[
\Pr\!\left[\,|X-\E[X]|>\eps\,\E[X]+\eps\,|S|\,\right]
\ \le\ \Pr\!\left[\,|X-\E[X]|>\eps\,\E[X]\,\right]
\ \le\ 2 e^{-\eps^{2}\,\E[X]/33}
\ \le\ e^{-\eps^{2}\,\opt/66},
\]
where we used the tail bound of the minimum vertex cover (Corollary~\ref{cor:mvc-concentration}) for $X$ and the fact that $\E[X] \ge \opt/2$.

\vspace{1em}
\noindent\textbf{Case (b): }$\E[X] \le \opt/2$.
Since $g(S)=|S|+X$ and $S\cup \MVC{\Gstar[V\setminus S]}$ is a vertex cover of $\Gstar$,
we have $g(S)\ge |\MVC{\Gstar}|$ for every realization, hence $\E[g(S)]\ge \opt$ and thus
$|S|+\E[X]=\E[g(S)]\ge \opt$.
Therefore, from \eqref{eq:split},
\begin{align*}
\Pr\!\left[\,|X-\E[X]|>\eps\,\E[X]+\eps\,|S|\,\right]
\ &\le\ \Pr\!\left[\,|X-\E[X]|>\eps\,\opt\,\right]\\
\ &=\ \Pr\!\left[\,|X-\E[X]|>\tfrac{\eps\,\opt}{\E[X]}\cdot \E[X]\,\right]\\
\ &\le\ 2 e^{-\eps^{2}\,\opt^{2}/(33\cdot\E[X])}\\
\ &\le\ 2 e^{-\eps^{2}\,\opt/66},
\end{align*}
where we again used concentration for $X$ and the bound $\E[X]\le \opt/2$.

\vspace{1em}

\noindent Combining the two cases yields the stated inequality.
\end{proof}

We are now ready to prove that, for any fixed realization $\Fstar$, $\optS$ is almost as good as adaptively selecting the best set among $\{\SEED(Q_{VC}, \Fstar)\;:\;Q_{VC}\subseteq Q\}$. We state this formally in the next lemma.

\begin{lemma}\label{lem:s-is-better-than-min-seed}
Let $\optS$ be an optimal solution to~\eqref{problem:best-S-with-Q}. Then, for any realization $\Fstar$,
\[
\E_{\Gstar\setminus \Fstar}\!\left[g\!\left(\optS\right)\right]
\;\le\; (1+O(\eps))\;\E_{\Gstar\setminus\Fstar}\!\left[\min_{Q_{VC}\subseteq Q}\, g\!\left(\SEED(Q_{VC}, \Fstar)\right)\right].
\]
\end{lemma}

\begin{proof}[Proof of Lemma~\ref{lem:s-is-better-than-min-seed}]
As we previously discussed, because $Q\subseteq S$, edges in $F$ never appear in $G[V\setminus S]$, therefore the objective of~\eqref{problem:best-S-with-Q} depends only on randomness outside $F$ and equals $\E_{\Gstar\setminus\Fstar}[g(S)]$.

\medskip
\noindent For any fixed $\Fstar$ and any $Q_{VC}\subseteq Q$, the set $S=\SEED(Q_{VC},\Fstar)$ is feasible for~\eqref{problem:best-S-with-Q}: (a) it contains $Q$ by definition, and (b) by Lemma~\ref{lem:v-minus-seed-is-sparse}, the graph $G[V\setminus \SEED(Q_{VC},\Fstar)]$ has $O(n/(\eps^5 p))$ edges. Hence, by the optimality of $\optS$,
\[
\E_{\Gstar\setminus\Fstar}\!\left[g\!\left(\optS\right)\right]
\;\le\; \min_{Q_{VC}\subseteq Q}\;\E_{\Gstar\setminus\Fstar}\!\left[g\!\left(\SEED(Q_{VC}, \Fstar)\right)\right].
\]

\noindent It remains to relate the expected cost of the best $\SEED$ to the expected cost of adaptively selecting the best $\SEED$ for every realization, i.e. the minimum expectation and the expected minimum. Define
\[
\mu \;:=\; \min_{Q_{VC}\subseteq Q}\;\E_{\Gstar\setminus\Fstar}\!\left[g\!\left(\SEED(Q_{VC}, \Fstar)\right)\right].
\]
By a union bound and Lemma~\ref{lem:concentration-g(s)} applied to each fixed $Q_{VC}$,
\[
\Pr\!\left[\min_{Q_{VC}\subseteq Q} g\!\left(\SEED(Q_{VC}, \Fstar)\right) < (1-\eps)\mu\right]
\;\le\; 2^{|Q|+1}\,e^{-\eps^2\;\opt/66}.
\]
Together with Claim~\ref{cl:Q} (which gives $|Q|\le \eps^2\opt/100$), we obtain
\[
\Pr\!\left[\min_{Q_{VC}\subseteq Q} g\!\left(\SEED(Q_{VC}, \Fstar)\right) < (1-\eps)\mu\right]
\;\le\;2\;\left(\frac{2}{e}\right)^{\eps^2 \opt/100}.
\]
Therefore,
\begin{align*}
\E_{\Gstar\setminus\Fstar}\!\left[\min_{Q_{VC}\subseteq Q} g\!\left(\SEED(Q_{VC}, \Fstar)\right)\right]
&\;\ge\; \textrm{Pr}\!\left[\min_{Q_{VC}\subseteq Q} g\!\left(\SEED(Q_{VC}, \Fstar)\right) \geq \,(1-\eps)\,\mu\right]\,(1-\eps)\,\mu\\
&\;\ge\; \Bigl(1-2\;\left(\tfrac{2}{e}\right)^{\eps^2 \opt/100}\Bigr)\,(1-\eps)\,\mu,
\end{align*}
which rearranges to
\[
\min_{Q_{VC}\subseteq Q}\;\E_{\Gstar\setminus\Fstar}\!\left[g\!\left(\SEED(Q_{VC}, \Fstar)\right)\right]
\;\le\; \frac{1}{\bigl(1-2\;(2/e)^{\eps^2 \opt/100}\bigr)(1-\eps)}\;
\E_{\Gstar\setminus\Fstar}\!\left[\min_{Q_{VC}\subseteq Q} g\!\left(\SEED(Q_{VC}, \Fstar)\right)\right].
\]
Using $\opt \ge C \log(1/\eps)/\eps^2$ to absorb the $2\cdot(2/e)^{\eps^2 \opt/2}$ term into $O(\eps)$ yields the claimed $(1+O(\eps))$ factor. As we have discussed, in the case of $\opt = O(\log(1/\eps)/\eps^2)$ the problem becomes trivial, in the sense that it is feasible to query the whole base graph $G$ (see Appendix~\ref{app:proofs-sec-4}).
\end{proof}

Finally, averaging the above result over the randomness of $\Fstar$ yields our final bound.

\begin{lemma}\label{lem:s-with-q-is-1-plus-epsilon}
Let $\optS$ be an optimal solution to~\eqref{problem:best-S-with-Q}. Then
\[
\E\!\left[g\!\left(\optS\right)\right]\ \le\ (1+O(\eps))\cdot \opt.
\]
\end{lemma}

\begin{proof}[Proof of Lemma~\ref{lem:s-with-q-is-1-plus-epsilon}]
From the law of total expectation, with respect to the edges $F$, we have that
\[
\E\!\left[g\!\left(\optS\right)\right]
\;=\; \E_{\Fstar}\!\left[\,\E_{\Gstar\setminus\Fstar}\!\left[g\!\left(\optS\right)\right]\right].
\]
By averaging the result of Lemma~\ref{lem:s-is-better-than-min-seed} over the randomness in $\Fstar$, we get
\[
\E\!\left[g\!\left(\optS\right)\right]
\;\le\; (1+O(\eps))\cdot
\E_{\Fstar}\!\left[\E_{\Gstar\setminus\Fstar}\!\left[\min_{Q_{VC}\subseteq Q}\, g\!\left(\SEED(Q_{VC},\Fstar)\right)\right]\right].
\]

Observe that the expectation in the right hand side above, is simply taken over the all randomness in $\Gstar$.
We can upper bound this term by picking $\Fstar$ to be the realization of edges in $F$ and $Q_{VC} = Q \cap \MVC{\Gstar}$,
which yields
\begin{align*}
\E\left[\min_{Q_{VC}\subseteq Q} g\!\left(\SEED(Q_{VC}, \Fstar)\right)\right]
&\;\le\; \E\left[g\!\left(\SEED(\Gstar)\right)\right]\\
&\;=\; \E\left[\bigl|\,\SEED(\Gstar) \cup \MVC{\Gstar[V\setminus \SEED(\Gstar)]}\,\bigr|\right]\\
&\;\le\; \E\left[\bigl|\,\SEED(\Gstar) \cup \MVC{\Gstar}\,\bigr|\right].
\end{align*}

The proof is concluded by employing Lemma~\ref{lem:seed-value-is-1-plus-epsilon}, 
\begin{align*}
\E\!\left[\,\bigl|\,\SEED(\Gstar) \cup \MVC{\Gstar}\,\bigr|\,\right]
&\;\le\; \E\!\left[\,\bigl|\MVC{\Gstar}\bigr|\,\right] \;+\; \E\!\left[\,\bigl|\,\SEED(\Gstar)\setminus \MVC{\Gstar}\,\bigr|\,\right]\\
&\;\le\; (1+O(\eps))\cdot\E\!\left[\,\bigl|\MVC{\Gstar}\bigr|\,\right].\qedhere
\end{align*}
\end{proof}

\noindent Finally, Theorem~\ref{thm:ouralg-1-plus-epsilon} follows from Lemma~\ref{lem:s-with-q-is-1-plus-epsilon} and the fact that Problem~\eqref{problem:best-S-with-Q} is a more constrained version of Problem~\eqref{problem:best-S}.

\section{Concentration Bound for Minimum Vertex Cover Size}\label{sec:mvc-concentration}In this section we prove that the size of the minimum vertex cover is concentrated around its expectation.  Rather well-known  edge-exposure and vertex-exposure Doob martingale arguments yield concentration bounds in terms of the number of edges and vertices in the graph, respectively.   We strengthen this by giving a  novel concentration bound in  $\opt$, where $\opt = \mathbb{E}[\abs{\MVC{\Gstar}}]$, yielding a multiplicative $(1\pm\varepsilon)$-type concentration guarantee.

We begin by defining a martingale, and stating a special case of Freedman's general inequality for martingales, which is an analogue of Bernstein's inequality for martingales. While Bernstein's inequality applies only to independent random variables having bounded variance, Freedman's inequality yields a similar bound for a sequence of random variables that is a martingale, and not necessarily independent.

\begin{definition}[Martingale]
    \label{def:martingale}
    A sequence of random variables $Y_0,\dots,Y_N$ is a \textit{martingale} with respect to a sequence of random variables $X_1,\dots,X_N$ if the following conditions hold:
    \begin{enumerate}
        \item $Y_k$ is a function of $X_1,\dots,X_k$ for all $k \ge 1$.
        \item $\E[|Y_k|] < \infty$ for all $k \ge 0$.
        \item $\E[Y_k|X_1,\dots,X_{k-1}]=Y_{k-1}$ for all $k \ge 1$.
    \end{enumerate}
\end{definition}

The following is a special case of Freedman's inequality for martingales:
\begin{theorem}[Freedman's Inequality]
    \label{thm:freedman}
    Let $Y_0,\dots,Y_N$ be a martingale with respect to the sequence $X_1,\dots,X_N$. Suppose the differences $D_k = Y_k - Y_{k-1}$ satisfy $\abs{D_k} \leq R$ for any $k \ge 1$ always, and suppose $W = \sum_{k=1}^N \Var(D_k | X_1,\dots,X_{k-1}) \leq \sigma^2$ always. Then, for any $t \ge 0$,
    \[ \Pr[\abs{Y_N - Y_0} \geq t] \leq 2 \exp\paren{-\frac{t^2}{2(\sigma^2 + Rt/3)}}. \]
\end{theorem}

We are restate our concentration bound from the introduction:

\concentration*

We note that we set the value of $C$ to be the constant from Lemma \ref{lem:structural}.

\eat{
\begin{theorem}\label{thm:MVC-concentration}
    Let $Z = |\MVC{\Gstar}|$ and $\opt = \mathbb{E}_{\Gstar\sim G_p}[\abs{\MVC{\Gstar}}]$. %
    Then for any $t \geq 0$,
    \begin{equation} \label{eq:concentration_bernstein}
        \Pr[\abs{Z - \opt} \ge t] \leq 2 \exp\paren{-\frac{t^2}{4C\cdot\opt + 2t/3}},
    \end{equation}
    where $C < 8$ is the constant from Lemma \ref{lem:structural}.
\end{theorem}
}%

\begin{proof}[Proof of \Cref{thm:MVC-concentration}]
To prove this result, we will require the existence of a set $S \subseteq V$ of vertices, which satisfies that $|S|=O(\opt)$, and furthermore, the induced subgraph $G[V \setminus S]$ has at most $O(\opt/p)$ edges. We defer the construction of such a set satisfying these properties to Lemma \ref{lem:structural} in the next section; here, we assume its existence and proceed. %
Let $U = V\setminus S$ and $E_U = E(G[U])$. Let $m = |E_U|$ and $k= |S|$. Fix an arbitrary ordering $E_U=\{e_1,\dots,e_m\}$ of all edges in $E_U$, as well as an ordering $S=\{v_1,\dots,v_k\}$ of the vertices in $S$. 

We will expose the randomness of the graph realization $G^\star$ in two phases via a sequence of random variables $X_1,\dots,X_m$, followed by $X_{m+1},\dots,X_{m+k}$. %
In the first $m$ steps, we reveal, in order, the existence of every edge in $E_U$ in $\Gstar$. In the next $k$ steps, we reveal, in order, the edges adjacent to each vertex in $S$.

\begin{enumerate}
  \item \textbf{Phase 1 (edge exposure inside $U$).} For each $i=1,\dots,m$, reveal the
  indicator $X_{i}=\mathbf{1}\{e_i\in E(\Gstar)\}$. That is, $X_i$ describes if edge $e_i \in E_U$ is realized in $\Gstar$. %

  \item \textbf{Phase 2 (vertex exposure on $S$).}
  For $j=1,\dots,k$, let $X_{m+j} = (\mathbf{1}\{(v_j,v) \in E(\Gstar)\})_{v \in V \setminus \{v_j\}}$. That is, $X_{m+j}$ describes the neighborhood of vertex $v_j \in S$ in $\Gstar$.
\end{enumerate}

Let $N=m+k$. We will define the Doob martingale $Y_0,Y_1,\dots,Y_N$ with respect to the sequence of random variables $X_1,\dots,X_N$ described above as $Y_i=\E[Z | X_1,\dots,X_{i-1}]$ for $i \ge 0$. Note that $Y_0=\opt$ and $Y_N=\abs{\MVC{\Gstar}}$. We observe how this martingale is a hybridization of the standard edge exposure and vertex exposure martingales, where the sequence $X_1,\dots,X_N$ comprises entirely of either Phase 1 or Phase 2 variables from above.

Denote the martingale differences by $D_i = Y_i - Y_{i-1}$, for $i=1,\dots,N$. Our goal will be show that the conditions of Freedman's inequality (Theorem \ref{thm:freedman}) are satisfied.

\medskip
First, we will show that the martingale differences are bounded by $R=1$ always, i.e., $|D_i| \le 1$ for all $i$.
\begin{itemize}
  \item \textbf{Phase 1 ($1 \le i \le m$).}
  At step $i$ in Phase 1, we reveal the status of edge $e_i \in E_U$. Let $x_1,\dots,x_i$ denote any assignment of the random variables $X_1,\dots,X_i$. We will denote $\E[Z \mid X_1=x_1,\dots,X_i=x_i]$ simply as $\E[Z \mid x_1,\dots,x_i]$ for convenience. Then, observe that
\begin{align*}
    |D_i| &= |Y_i - Y_{i-1}| = \left|\E\left[Z | x_1,\dots,x_i\right]-\E\left[Z | x_1,\dots,x_{i-1}\right]\right| \\
    &= \left|\sum_{\tilde{x}_{i}}\Pr[X_i=\tilde{x}_i \mid  X_1=x_1,\dots,X_{i-1}=x_{i-1}] \cdot \left(\E\left[Z|x_1,\dots,x_i\right]-\E\left[Z | x_1,\dots,\tilde{x}_i\right] \right)\right| \\
    &\le \sum_{\tilde{x}_{i}}\Pr[X_i=\tilde{x}_i \mid X_1=x_1,\dots,X_{i-1}=x_{i-1}] \cdot \left| \left(\E\left[Z | x_1,\dots,x_i\right]-\E\left[Z | x_1,\dots,\tilde{x}_i\right] \right)\right|
    \end{align*}
  Now, let $R_{\rest}$ denote all of the remaining randomness in $\Gstar$ beyond $X_1,\dots,X_i$; since $R_{\rest}$ is independent of $X_1,\dots,X_i$, we have that
  \begin{align*}
    &\left|\E\left[Z | x_1,\dots,x_i\right]-\E\left[Z | x_1,\dots,\tilde{x}_i\right]\right| \\
    &= \left|\sum_{r_{\rest}}\Pr[R_{\rest}=r_{\rest}]\cdot \left(\E[Z | x_1,\dots,x_i,r_{\rest}]-\E[Z \mid x_1,\dots,\tilde{x}_i,r_{\rest}]\right)\right|.
  \end{align*}
  But now, note that the graph $\Gstar$ is completely determined by $x_1,\dots,x_{i},r_\rest$ (respectively $x_1,\dots,\tilde{x}_{i},r_\rest$); furthermore, these graphs differ only in the realization of the edge $X_i$; in particular, the size of the minimum vertex covers of these two graphs can differ by at most 1. Substituting this above, we get that $|D_i| \le 1$.

  \item \textbf{Phase 2 ($m+1 \le i \le N$).}
  Let $j=i-m$. At Step $i$ in Phase 2, we reveal all edges incident to vertex $v_j \in S$, namely $X_i = (\textbf{1}\{(v_j,v) \in E(\Gstar)\})_{v \in V \setminus \{v_j\}}$. Similar to the calculation above, where $x_i$ now denotes an assignment to all the indicator random variables involved in $X_i$, we have that
    \begin{align*}
    |D_i| \le \sum_{\tilde{x}_{i}}\Pr[X_i=\tilde{x}_i \mid X_1=x_1,\dots,X_{i-1}=x_{i-1}] \cdot \left| \left(\E\left[Z|x_1,\dots,x_i\right]-\E\left[Z|x_1,\dots,\tilde{x}_i\right] \right)\right|
    \end{align*}
  Now, let $R_{\rest}$ denote all of the remaining randomness in $\Gstar$ beyond $X_1,\dots,X_i$; since $R_{\rest}$ is independent of $X_1,\dots,X_i$, we have that
  \begin{align*}
&\left|\E\left[Z|x_1,\dots,x_i\right]-\E\left[Z|x_1,\dots,\tilde{x}_i\right]\right| \\
    &= \left|\sum_{r_{\rest}}\Pr[R_{\rest}=r_{\rest}]\cdot \left(\E[Z|x_1,\dots,x_i,r_{\rest}]-\E[Z|x_1,\dots,\tilde{x}_i,r_{\rest}]\right)\right|.
  \end{align*}
  But now, note that the graph $\Gstar$ is completely determined by $x_1,\dots,x_{i},r_\rest$ (respectively $x_1,\dots,\tilde{x}_{i},r_\rest$); furthermore, these graphs differ only in the realization of the neighborhood of the single vertex $v_j=v_{i-m}$; in particular, the size of the minimum vertex covers of these two graphs can differ by at most 1 (either $v_j$ is included in one or not). Substituting this above, we get that $|D_i| \le 1$.
  
\end{itemize}

Summarily, we conclude that $|D_i| \le R = 1$ for every $i \ge 1$ always. 

Next, we bound the variance $W = \sum_{i=1}^N \Var(D_i | X_1,\dots,X_{i-1})$. We will split it into $W_1$ (Phase 1) and $W_2$ (Phase 2).

\begin{itemize}
    \item{\bf Bounding $W_1 = \sum_{i=1}^m \Var(D_i | X_1,\dots,X_{i-1})$.} Consider any fixing of the random variables $X_1=x_1,\dots,X_{i-1}=x_{i-1}$, and recall that 
    \begin{align*}
    D_i = Y_i - Y_{i-1} = \E\left[Z|x_1,\dots,x_{i-1},X_i\right]-\E\left[Z|x_1,\dots,x_{i-1}\right] 
    \end{align*}
    Note that $X_i$ is independent of $X_1,\dots,X_{i-1}$, and is in fact a Bernoulli random variable with parameter $p$. In particular, conditioned on $x_1,\dots,x_{i-1}$, the distribution of $D_i$ is as follows:
    \begin{align*}
        D_i = \begin{cases}
            \E\left[Z|x_1,\dots,x_{i-1},X_i=1\right]-\E\left[Z|x_1,\dots,x_{i-1}\right] := a & \text{with probability $p$} \\
            \E\left[Z|x_1,\dots,x_{i-1},X_i=0\right]-\E\left[Z|x_1,\dots,x_{i-1}\right] := b & \text{with probability $1-p$}.
        \end{cases}
    \end{align*}
    Therefore, we have that
    \begin{align*}
        \Var(D_i | x_1,\dots,x_{i-1}) &= pa^2+(1-p)b^2-(pa+(1-p)b)^2 %
        = p(1-p)(a-b)^2 \\
        &= p(1-p)|\E\left[Z|x_1,\dots,x_{i-1},X_i=1\right]-\E\left[Z|x_1,\dots,x_{i-1},X_i=0\right]|^2.
    \end{align*}
    By similar reasoning as earlier in the proof, where we realize the rest of the randomness, we have that 
    $$|\E\left[Z|x_1,\dots,x_{i-1},X_i=1\right]-\E\left[Z|x_1,\dots,x_{i-1},X_i=0\right]|^2 \le 1.$$
    So, we get finally that $W_1 = \sum_{i=1}^m \Var(D_i | x_1,\dots,x_{i-1}) \le mp(1-p) =p(1-p)|E_U|$. By Lemma \ref{lem:structural}, $|E_U| \leq C \cdot \opt/p$. So, $W_1 \le (1-p) C \cdot \opt \leq C \cdot \opt$.

    \item{\bf Bounding $W_2 = \sum_{i=m+1}^{N} \Var(D_i | X_1,\dots,X_{i-1})$.} In this case, note that by our previous calculations, we know that $-1 \le D_i \le 1$ always, which means that $\Var(D_i | X_1,\dots,X_{i-1}) \le 1$. This implies that $\sum_{i=m+1}^{N} \Var(D_i | X_1,\dots,X_{i-1}) \le k \le C \cdot \opt$, where we used the guarantee of Lemma \ref{lem:structural} again.
    
\end{itemize}
In total, $W = W_1 + W_2 \le 2 C\cdot\opt$ always.
Finally, we apply Freedman's inequality with $R=1$ and $\sigma^2 = 2C\cdot\opt$, to get
\begin{align*}
    \Pr[|Z - \opt| \ge t] \le 2\exp\left(-\frac{t^2}{4C \cdot \opt + 2t/3}\right),
\end{align*}
which completes the proof.
\end{proof}
\begin{corollary}\label{cor:mvc-concentration}
Let $Z = |\MVC{\Gstar}|$ and $\opt = \mathbb{E}[\abs{\MVC{\Gstar}}]$. 
For any $t\le \opt$
    \begin{equation} \label{eq:concentration_gaussian}
        \Pr[|Z - \opt| > t] \leq 2 e^{- (t^2/33)/\opt}.
    \end{equation}
\end{corollary}
\begin{proof}
Apply Theorem~\ref{thm:MVC-concentration}. If $t \leq \opt$, then the denominator is upper bounded by $4C\cdot\opt + 2\opt/3 = (4C+2/3)\opt$. This finally gives
    \[ \Pr[\abs{Z - \opt} \ge t] \leq 2 \exp\paren{-\frac{t^2}{(4C+2/3)\cdot\opt}} %
    \]
    This is of the form $C' e^{-c t^2/\opt}$ with $C'=2$ and $c = 1/(4C+2/3) \geq 1/33$.
\end{proof}

\subsection{A Structural Lemma based on Expected Minimum Vertex Cover}
Next, we establish a structural lemma on minimum vertex cover, which is a crucial component in the concentration result of Theorem~\ref{thm:MVC-concentration}.
\begin{lemma}
    \label{lem:structural}
    Let $G=(V,E)$ be a graph, and $p \in (0, 1]$. Let $\opt = \E[\abs{\MVC{\Gstar}}]$. There exists a subset of vertices $S \subseteq V$ such that:
    \begin{enumerate}
        \item $\abs{S} \leq c \cdot \opt$.
        \item The induced subgraph $G[V\setminus S]$ has at most $C \frac{\opt}{p}$ edges.
    \end{enumerate}
    The constants can be chosen as $c = C = 2\paren{\frac{e}{e-1} + 2} < 8$.
\end{lemma}

\begin{proof}
    We will first construct a specific ordering $\pi$ of the vertices in $G$, by analyzing the following 
    greedy process, which incrementally constructs $\pi$.
   
    \medskip 
    \noindent Initialize $P = ()$.
    While $V \setminus P \neq \emptyset$:
    \begin{enumerate}
        \item Suppose $P=(v_1,\dots,v_{|P|})$ so far. Draw a random realization $G^\star \sim G_p$ %
        \item For $i=1,2,\dots,|P|$:
        \begin{itemize}[label=--]
            \item If $v_i$ is unmatched, match $v_i$ uniformly at random to one of the yet-unmatched neighbors of $v_i$ within the set $V \setminus \{v_1,\dots,v_{i-1}\}$ in the realization $G^\star$. 
        \end{itemize}
        \item The above process induces a probability distribution for every vertex in $V \setminus P$ to be matched to a vertex in $P$, over the randomness in the draw of $G^\star$ and the matching---so, for every vertex $v \in V \setminus P$, denote by $p(v|P)$ the probability that it gets matched to a vertex in $P$.
        \item Select $u = \arg\max_{v \in V \setminus P} p(v|P)$.
        \item Append $u$ to $P$.
    \end{enumerate}

    Let $\pi$ be the resulting permutation of the vertices in $G$ at the end of the procedure above. For $\Gstar \sim G_p$, running only the for-loop in Step 2 above with $P=\pi$ constitutes a greedy procedure for constructing a matching $M$ of $G^\star$. We will analyze the matching $M$ that results by this process in the following. Note that $\E[\abs{M}] \leq \E[\abs{\MM{\Gstar}}] \leq \opt$, where $\MM{\Gstar}$ is a \textit{maximum} matching of $\Gstar$, whose size is at most $\abs{\MVC{\Gstar}}$.

    \smallskip
    For any vertex $u$, let $P_u$ be its predecessors in $\pi$ and let $S_u$ be its successors.
    \begin{itemize}
        \item $\delta^+(u)$: The number of neighbors of $u$ in $G$ that are in $S_u$ (forward degree).
        \item $p(u)$: The probability that $u$ is matched by a vertex in $P_u$ (matched backwards).
        \item $P(u)$: The total probability that $u$ is matched in $M$.
        \item $R(u) = 1 - e^{-p\delta^+(u)}$.
    \end{itemize}

    A crucial consequence of the greedy construction of $\pi$ is the following property: For any successor $v \in S_u$, the probability that $v$ is matched by $P_u$ is at most $p(u)$. This is because when $u$ was selected to be added to $\pi$, it had the maximum probability of being matched by $P_u$ among all remaining vertices, including $v$.
    We aim to prove the following key inequality:
    \begin{equation} \label{eq:key_inequality}
        P(u) \geq \max\paren{p(u), (1-2p(u)) R(u)}.
    \end{equation}

    Trivially, $P(u) \geq p(u)$, since $p(u)$ accounts for only part of the total probability $P(u)$. We now focus on the second term. If $p(u) \geq 1/2$, the term is non-positive and the inequality holds trivially. So, assume $p(u) < 1/2$.
    
    Let $A$ be the event that $u$ is matched backwards by $P_u$, so that $\Pr[A]=p(u)$.
    For $v \in S_u$, let $C_v$ be the event that $v$ is matched backwards to a vertex in $P_u$. As argued above, by the construction of $\pi$, $\Pr[C_v] \leq p(u)$.
    We now analyze the probability of the event $F$ that $u$ matches forward. Again, we have that $P(u) \geq \Pr[F]$. Observe that $u$ matches forward, only if $u$ is unmatched backwards ($A^c$), and there exists a successor $v$ such that the edge $(u,v)$ gets realized in $\Gstar$ AND $v$ is unmatched when $u$ is processed in the greedy procedure according to the permutation $\pi$ for constructing the matching $M$ of $\Gstar$ (i.e., the event $C_v^c$). %

    Now, to further analyze $\Pr[F]$, let us fix an arbitrary ordering on the successors $S_u = \{v_1, v_2, \dots\}$. Let $B_i$ be the event that $(u, v_i)$ is the first realized edge from $u$ to $S_u$ in this ordering. The events $B_i$ are disjoint. Let $B = \cup_i B_i$ be the event that at least one edge from $u$ to $S_u$ is realized. We have that $\Pr[B] = 1-(1-p)^{\delta^+(u)}$. Since $1-p \leq e^{-p}$, we have $\Pr[B] \geq 1-e^{-p\delta^+(u)} = R(u)$.

    Let $F_i$ be the event that $A^c$, $B_i$, and $C_{v_i}^c$ all occur. If any $F_i$ occurs, then there exists a realized edge $(u,v_i)$ where both $u$ and $v_i$ are available when $u$ is processed. Thus, the greedy matching procedure will match $u$ forward.
    Since the $B_i$'s are disjoint, the $F_i$'s are disjoint, which implies that
    \[ \Pr[F] \geq \Pr[\cup_i F_i] = \sum_i \Pr[A^c \cap B_i \cap C_{v_i}^c]. \]

    We can now use independence. The event $B_i$ depends only on the realization of edges between $u$ and $S_u$. The events $A$ and $C_{v_i}$ depend only on the realization of edges incident to $P_u$. Since these sets of edges are disjoint, $B_i$ is independent of the joint event $(A^c \cap C_{v_i}^c)$, giving us that
    \[ \Pr[F] \geq \sum_i \Pr[B_i] \Pr[A^c \cap C_{v_i}^c]. \]
    Next, we use the union bound to lower bound the probability that both $u$ and $v_i$ are available.
    \begin{align*}
        \Pr[A^c \cap C_{v_i}^c] = 1 - \Pr[A \cup C_{v_i}] \geq 1 - (\Pr[A] + \Pr[C_{v_i}]).
    \end{align*}
    Using $\Pr[A]=p(u)$ and the crucial property that $\Pr[C_{v_i}] \leq p(u)$, we get that
    \[ \PP[A^c \cap C_{v_i}^c] \geq 1 - 2p(u). \]
    Substituting this back into the expression for $\Pr[F]$ finally gives that
    \begin{align*}
        \Pr[F] \geq \sum_i \Pr[B_i] (1-2p(u)) = (1-2p(u)) \sum_i \Pr[B_i] = (1-2p(u)) \Pr[B].
    \end{align*}
    Since $\Pr[B] \geq R(u)$, we conclude $P(u) \geq \Pr[F] \geq (1-2p(u))R(u)$. This establishes the desired inequality (\ref{eq:key_inequality}). We now proceed towards constructing the set $S$ from the lemma statement.

    \paragraph{Construction of the Set $S$.}
    Define the set $S = \{u \in V \mid p\delta^+(u) \geq 1\}$.
    Let $K = 1-e^{-1}$. %
    For $u \in S$, $p\delta^+(u) \geq 1$, so $R(u) \geq 1-e^{-1} = K$.
    Applying this to Inequality (\ref{eq:key_inequality}):
    \[ P(u) \geq \max(p(u), (1-2p(u))K). \]
    This expression is minimized when the two terms are equal: $p(u) = (1-2p(u))K$. This yields $p(u) = K/(1+2K)$.
    Let $C_0 = \frac{K}{1+2K}$. %
    Thus, for all $u \in S$, $P(u) \geq C_0$.
    Next, we bound $\abs{S}$ using the expected size of the matching $M$:
    \[ \E[\abs{M}] = \frac{1}{2} \sum_{u \in V} P(u). \]
    Since $\E[\abs{M}] \leq \opt$, we have $\sum_{u \in V} P(u) \leq 2\cdot \opt$. So,
    \[ 2\cdot \opt \geq \sum_{u \in S} P(u) \geq C_0 \abs{S}. \]
    Therefore, $\abs{S} \leq \frac{2}{C_0} \opt$. This proves the first part of the lemma with $c = \frac{2}{C_0}$.
    Note that since $\frac{1}{C_0} = \frac{1+2K}{K} = \frac{e}{e-1} + 2$, we have $c = 2\paren{\frac{e}{e-1} + 2}$.

    \paragraph{Bounding Edges in $G[V\setminus S]$.}
    Let $U = V\setminus S$. For $u \in U$, we have $p\delta^+(u) < 1$.
    We use the inequality that for $x \in [0, 1]$, $1-e^{-x} \geq Kx$ (due to the concavity of $1-e^{-x}$).
    Thus, $R(u) \geq K p\delta^+(u)$.
    Applying this to Inequality (\ref{eq:key_inequality}),
    \[ P(u) \geq \max\paren{p(u), (1-2p(u)) K p\delta^+(u)}. \]
    Similar to the minimization of $P(u)$, computed in the construction of $S$, this yields a lower bound:
    \[ P(u) \geq \frac{K p\delta^+(u)}{1+2K p\delta^+(u)}. \]
    We seek a linear lower bound $C' p\delta^+(u)$. We minimize the ratio $\frac{K}{1+2K p\delta^+(u)}$. Since $p\delta^+(u) < 1$, the ratio is minimized as $p\delta^+(u) \to 1$. The minimum ratio is $\frac{K}{1+2K} = C_0$.
    So, for all $u \in U$, $P(u) \geq C_0 p\delta^+(u)$.
    Now we sum this probability over $U$:
    \[ 2\cdot \opt \geq \sum_{u \in U} P(u) \geq \sum_{u \in U} C_0 p\delta^+(u) = C_0 p \sum_{u \in U} \delta^+(u). \]

    Let $\abs{E(G[U])}$ be the number of edges in the induced subgraph $G[U]$. Every edge $(u,v)$ in $G[U]$ has both endpoints in $U$. In the permutation $\pi$, one vertex must precede the other. The edge is counted exactly once in the forward degree of the earlier endpoint. Thus,
    \[ \abs{E(G[U])} \le \sum_{u \in U} \delta^+(u) \leq \frac{2}{C_0} \cdot \frac{\opt}{p}. \]
    This proves the second part of the lemma with $C = \frac{2}{C_0} = 2\paren{\frac{e}{e-1} + 2}$.
\end{proof}

\begin{remark}[A Talagrand-based alternative]
\label{rem:talagrand-mvc}
One can also derive a weaker concentration bound for $\opt = \mathbb{E}[\abs{\MVC{\Gstar}}]$ via
Talagrand's inequality on the vertex-exposure.

Fix an arbitrary ordering $V=\{v_1,\dots,v_n\}$ of the vertices of $G$. For each
$i\in[n]$, let
\[
X_i = \bigl(\mathbf{1}\{(v_i,v_j)\in E(\Gstar)\}\bigr)_{j<i},
\]
so that $X_i$ records the realized edges from $v_i$ to the vertices preceding it
in the ordering. Since edges are realized independently, 
$X_1,\dots,X_n$ are mutually independent, and there exists a function
$f$ such that
\[
\abs{\MVC{\Gstar}} = f(X_1,\dots,X_n).
\]
As discussed in this section, the function $f$ is $1$-Lipschitz with respect to the Hamming metric on the
coordinates: changing a single coordinate $X_i$ only changes the realized edges
incident to the vertex $v_i$, and therefore can change the minimum vertex-cover
size by at most $1$. 

Moreover, the event $\{f(\cdot) \ge r\}$ is $2$-certifiable.
Indeed, if $\abs{\MVC{\Gstar}} \ge r$, then there exists a subgraph
$H\subseteq G^\star$ that is minimal subject to $\abs{\MVC{H}}\ge r$.
The classical bound by Erd\H{o}s and Gallai for such graphs implies that $\abs{V(H)}\le 2r$; see, e.g.,
\cite[Chapter 12.1]{lovasz2009matching} or \cite{gyarfas2021order,erdHos1961minimal}.
Since $H$ is inclusion-minimal with $\abs{\MVC{H}}\ge r$, we in fact have $\abs{\MVC{H}}=r$,
and therefore revealing the coordinates of the
vertices of $H$ certifies the event. In other words, $f$ is $2$-certifiable. 

Applying Talagrand's inequality for $c$-Lipschitz $s$-certifiable functions; see, e.g.,~\cite[Chapter 16.2]{molloy2002graph}, there exists absolute constants
$C,c>0$ such that for all $t\ge 0$,
\[
\Pr\bigl(\bigl| \abs{\MVC{\Gstar}} - \opt \bigr| \ge t + C\sqrt{\opt}\bigr)
    \le 4\exp\left(-\frac{c t^2}{\opt}\right).
\]

By a standard two-case argument, it also implies that for some absolute constant $c'>0$,
\[
\Pr\left(\bigl|\abs{\MVC{\Gstar}}-\opt\bigr|>t\right)
   \le 2\exp\left(-\frac{c't^2}{\opt}\right)
   \qquad\text{for all } t\in[0,\opt],
\]
which matches the grantee of \Cref{cor:mvc-concentration} up to a constant factor.
\end{remark}

\section{Concluding Remarks}In this paper, we gave the first algorithm for the stochastic vertex problem that obtain a $(1+\eps)$-approximation using $O_\eps(n/p)$ queries. This result is optimal, up to the dependence on $\eps$ in $O_\eps(\cdot)$, which for our algorithm is $1/\eps^5$. As in prior work, the set of edge queries in our algorithm is non-adaptive. Our algorithm is simple, and in some sense, the canonical algorithm for non-adaptive queries. A key tool in our analysis is a new concentration bound on the size of minimum vertex cover in random graphs, which might be of independent interest. While this result concludes the search for better stochastic vertex cover algorithms (modulo sharpening the dependence of the query bound on $\eps$), we believe that stochastic approximation for fundamental graph problems is an interesting domain for future research that is currently under-explored.\label{sec:conclusion}

\section*{Acknowledgments}
The authors thank the anonymous reviewer for pointing out the alternative approach in~\Cref{rem:talagrand-mvc} for proving the concentration bound for the size of a minimum vertex cover.

Miltiadis Stouras and Ola Svensson are supported by the Swiss State Secretariat for Education, Research and Innovation (SERI) under contract number MB22.00054.
Jan van den Brand is supported by NSF Award CCF-2338816 and CCF-2504994. 
Inge Li Gørtz is supported by Danish Research Council grant DFF-8021-002498.
Chirag Pabbaraju is supported by Gregory Valiant's and Moses Charikar's Simons Investigator Awards, and a Google PhD Fellowship.
Cliff Stein is supported in part  by NSF grant CCF-2218677, ONR grant ONR-13533312, and by the Wai T. Chang Chair in Industrial Engineering and Operations Research. Debmalya Panigrahi is supported in part by NSF grants CCF-2329230 and CCF-1955703. 

\bibliographystyle{alpha}
\bibliography{mvc-ref}

\appendix

\section{Handling the case of \texorpdfstring{$\opt = O\bigl(\log(1/\eps)/\eps^2\bigr)$}{}}\label{app:proofs-sec-4}\begin{lemma}\label{lem:opt-large-suffices}
For any universal constant \(C'>0\), if
\[
\opt \;<\; \frac{C'\,\log(1/\eps)}{\eps^2},
\]
then the total number of edges of \(G\) is \(O\!\bigl(n\,\log(1/\eps)/\eps^2\bigr)\). In particular, in this regime the algorithm can query the entire graph and return an exact solution. Consequently, it suffices to analyze the case
\[
\opt \;\ge\; \frac{C'\,\log(1/\eps)}{\eps^2}.
\]
\end{lemma}

\begin{proof}[Proof of Lemma~\ref{lem:opt-large-suffices}]
Suppose that $\opt \;<\; \frac{C'\,\log(1/\eps)}{\eps^2}$. By the Lemma~\ref{lem:structural}, there exists a set \(S\subseteq V\) with \(|S|=\Theta(\opt)\) such that the induced subgraph \(G[V\setminus S]\) has \(O(\opt)\) edges. Every remaining edge of \(G\) is incident to at least one vertex of \(S\), hence the number of such edges is at most \(n\cdot |S|=O(n\cdot\opt)\). Therefore
\[
|E(G)| \;\le\; n|S| + O(\opt) \;=\; O(n\cdot\opt) \;=\; O\!\Bigl(n\,\frac{\log(1/\eps)}{\eps^2}\Bigr).\qedhere
\]
\end{proof}

\end{document}